\documentclass{article}

\usepackage[final]{neurips_2019}

\usepackage[utf8]{inputenc} % allow utf-8 input
\usepackage[T1]{fontenc}    % use 8-bit T1 fonts
\usepackage{hyperref}       % hyperlinks
\usepackage{url}            % simple URL typesetting
\usepackage{booktabs}       % professional-quality tables
\usepackage{amsfonts}       % blackboard math symbols
\usepackage{nicefrac}       % compact symbols for 1/2, etc.
\usepackage{microtype}      % microtypography
\usepackage{amsmath, amsthm, amssymb}
\usepackage{graphicx}
\usepackage{stefan_tex}
\usepackage{sidecap}

\hypersetup{colorlinks,citecolor=blue,urlcolor=blue,linkcolor=blue}

\usepackage{natbib}

\usepackage{caption}

\usepackage{floatrow}
\newfloatcommand{capbtabbox}{table}[][\FBwidth]

\theoremstyle{plain}
\newtheorem{prop}{Proposition}

\newtheorem{coro}[prop]{Corollary}
\newtheorem{lemm}[prop]{Lemma}
\newtheorem{theo}[prop]{Theorem}

\theoremstyle{definition}

\theoremstyle{remark}

\newtheorem{rema}[prop]{Remark}

\title{Covariate-Powered Empirical Bayes Estimation}

\usepackage{changes}

\author{%
  Nikolaos Ignatiadis\\
  Statistics Department\\
  Stanford University\\
  \texttt{ignat@stanford.edu} \\
  \And
   Stefan Wager \\
   Graduate School of Business \\
   Stanford University \\
  \texttt{swager@stanford.edu} \\
}

\begin{document}

\maketitle

\begin{abstract}
We study methods for simultaneous analysis of many noisy experiments in the presence
of rich covariate information. The goal of the analyst is to optimally estimate the true effect underlying
each experiment. Both the noisy experimental results and the auxiliary covariates are useful for this purpose,
but neither data source on its own captures all the information available to the analyst.
In this paper, we propose a flexible plug-in empirical Bayes estimator that synthesizes both sources of information and may leverage any black-box predictive model. We show that our approach is within a constant factor of minimax for a simple data-generating model.
Furthermore, we establish robust convergence guarantees for our method that hold under considerable
generality, and exhibit promising empirical performance on both real and simulated data.
\end{abstract}

\section{Introduction}

It is nowadays common for a geneticist to simultaneously study the association
of thousands of different genes with a disease \citep{efron2001empirical, lonnstedt2002replicated, love2014moderated},
for a technology firm to have records from thousands of randomized experiments \citep{mcmahan2013ad},
or for a social scientist to examine data from hundreds of different regions at once \citep{abadie2018choosing}.
In all of these settings, we are fundamentally interested in learning something about each sample
(i.e., gene, experimental intervention, etc.) on its own; however, the abundance of data on other samples
can give us useful context with which to interpret our measurements about each individual sample
\citep{efron2010large,robbins1964empirical}.
In this paper, we propose a method for simultaneous analysis of many noisy experiments, and show that it is able to
exploit rich covariate information for improved power by leveraging existing machine learning tools
geared towards a basic prediction task.

As a motivation for our statistical setting, suppose we have access to a dataset of movie reviews where each
movie $i = 1, \, ..., \, n$ has an average rating $Z_i$ over a limited number of viewers; we also have
access to a number of covariates $X_i$ about the movie (e.g., genre, length, cast, etc.).
The task is to estimate the ``true'' rating $\mu_i$ of the movie, i.e., the average rating had the
movie been reviewed by a large number of reviewers similar to the ones who already reviewed it.
A first simple approach to estimating $\mu_i$ is to use its observed average rating
as a point estimate, i.e., to set $\hmu_i = Z_i$. This approach is clearly valid for movies
where we have enough data for sampling noise to dissipate, e.g., with over 50,000 reviews in the MovieLens 20M
data \citep{harper2016movielens}, we expect the 4.2/5 rating of Pulp Fiction to be quite stable. 
Conversely, for movies with fewer reviews, this strategy may be unstable: the rating 1.6/5 of Urban Justice
is based on less than 20 reviews, and appears liable to change as we collect more data. A second alternative
would be to just rely on covariates: We could learn to predict average ratings from covariates,
$m(x) = \EE{Z_i \cond X_i = x}$, and then set $\hmu_i = \hatm(X_i)$. This may be more appropriate
than using the observed mean rating for movies with very few reviews, but is limited in its accuracy
if the covariates aren't expressive enough to perfectly capture $\mu_i$.

We develop an approach that reconciles (and optimally interpolates between) the two estimation
strategies discussed above. The starting point for our discussion is the following generative model,
\begin{equation}
\label{eq:np-fay-herriot-model}
X_i \sim \mathbb{P}^X,\;\; \mu_i \mid X_i \sim \nn\p{ m(X_i), \,  A},\;\; Z_i \mid \mu_i \sim \nn\p{\mu_i, \,\sigma^2},
\end{equation}
according to which the true rating $\mu_i$ of each movie is partially explained by its covariates
$X_i$, but also has an idiosyncratic and unpredictable component with a Gaussian distribution $\nn\p{0, \, A}$.
Recall that we observe $X_i$ and $Z_i$ for each $i = 1, \, ..., \, n$, and want to estimate the vector of
$\mu_i$. Given this setting, if we knew both the idiosyncratic noise level $A$ and $m(x)$, the conditional mean of $\mu_i$ given $X_i = x$,
then the mean-square-error-optimal estimate of $\mu_i$ could directly be read off of Bayes' rule,
$\hmu_i^* = t^*_{m,A}(X_i, \, Z_i)$, with
\begin{equation}
\label{eq:Bayes_rule}
t^*_{m,A}(x,z) := \EE[m,A]{\mu_i \mid X_i=x, \, Z_i=z} = \frac{A}{\sigma^2 + A} z + \frac{\sigma^2}{\sigma^2 + A} m(x).
\end{equation}
As shown in Figure \ref{fig:np_vs_fh}, the behavior of this shrinker depends largely on the ratio
$A/\sigma^2$: As this ratio gets large, the Bayes rule gets close to just setting $\hmu_i = Z_i$, whereas
when the ratio is small, it shrinks everything to predictions made using covariates.

\begin{figure}
    \centering
    \includegraphics[width=0.8\linewidth]{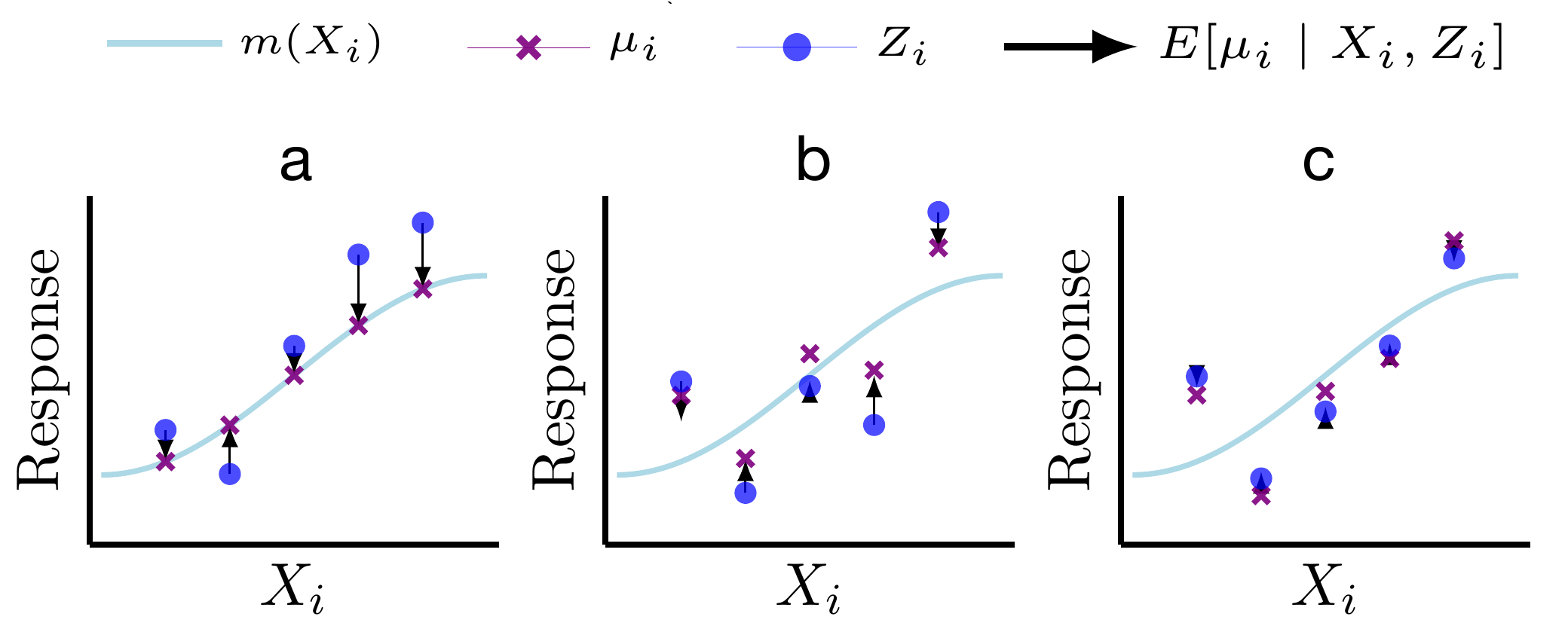}
    \caption{\textbf{Optimal empirical Bayes shrinkage.} All three plots show $\mu_i$ and $Z_i$ drawn from~\eqref{eq:np-fay-herriot-model} for various values of $A/\sigma^2$, with the covariate values $X_i$ fixed and the regression curve $m(\cdot)$ shown in blue.  The arrows depict how the oracle Bayes denoiser from~\eqref{eq:Bayes_rule} moves the point estimate \smash{$\hmu_i$} away from the raw observation $Z_i$ and towards $m(X_i)$. \textbf{a)} When $A/\sigma^2=0$, the oracle estimator shrinks $Z_i$ all the way back to $m(X_i)$. \textbf{b)} For $A/\sigma^2=1$, optimal shrinkage uses $(Z_i + m(X_i))/2$ to estimate $\mu_i$. \textbf{c)} When $A/\sigma^2$ is very large, it is preferable to discard $m(X_i)$ and just use the information in $Z_i$.}
    \label{fig:np_vs_fh}
\end{figure}

Now in practice, $m(\cdot)$ and $A$ are unlikely to be known a-priori and, furthermore, we
may not believe that the hierarchical structure \eqref{eq:np-fay-herriot-model} is a perfect description
of the underlying data-generating process. The main contribution of this paper is an estimation strategy
that addresses these challenges. First, we derive the minimax risk for estimating $\mu_i$ in model
\eqref{eq:np-fay-herriot-model} in a setting where $m(\cdot)$ is unknown but we are willing to make
various regularity assumptions (e.g., that $m(\cdot)$ is Lipschitz). Second, we show that a feasible
plug-in version of \eqref{eq:Bayes_rule} with estimated \smash{$\hatm(\cdot)$} and \smash{$\hA$}
attains this lower bound up to constants that do not depend on $\sigma^2$ or $A$.

Finally, we consider robustness of our approach to misspecification of the model
\eqref{eq:np-fay-herriot-model}, and establish an extension to the classic result of
\citet{james1961estimation}, whereby without any assumptions on the distribution of
$\mu_i$ conditionally on $X_i$, we can show that our approach still improves over
both simple baselines $\hmu_i = Z_i$ and $\hmu_i = \hatm(X_i)$ in considerable generality
(see Section \ref{sec:compound} for precise statements). We also consider behavior of our
estimator in situations where the distribution of $Z_i$ conditionally on $\mu_i, X_i$ may not be
Gaussian, and the conditional variance $\sigma^2_i$ of $Z_i$ given $\mu_i, X_i$ may be different
for different samples.

Our approach builds on a long tradition of empirical Bayes estimation
that seeks to establish frequentist guarantees for plug-in Bayesian estimators and related procedures
in data-rich environments \citep{efron2010large,robbins1964empirical}.
Empirical Bayes estimation in the setting without covariates $X_i$ is by now well understood
\citep{brown2009nonparametric,efron2011tweedie,efron1973stein, ignatiadis2019empirical, ignatiadis2019bias,james1961estimation, jiang2009general,johnstone2004needles,muralidharan2010empirical,stephens2016false, weinstein2018group}.

In contrast, empirical Bayes analysis with covariates has been less comprehensively explored,
and existing formal results are confined to special cases.
\citet{fay1979estimates} introduced a model of the form~\eqref{eq:np-fay-herriot-model}
with a linear specification, \smash{$m(x) = x^\top \beta$}, motivated by the problem of ``small area
estimation'' that arises when studying small groups of people based on census data.
Further properties of empirical Bayes estimators in the linear specification (including robustness
to misspecification) were established  by~\citet{green1991james} in the case where $X_i\in \RR$ and $m(x)=x$, 
and by~\citet{cohen2013empirical, tan2016steinized} and \citet{kou2017optimal} when \smash{$m(x)=x^\top \beta$}.
There has also been some work on empirical Bayes estimation with nonparametric specifications for $m$, e.g.,
\citet{mukhopadhyay2004two} and \citet{opsomer2008non}.  In a genetics application,~\citet{stephan2015random} 
parametrized $m(x)$ as a random forest. \citet{banerjee2018adaptive} utilize univariate side information to estimate sequences of $\mu_i$ that consist mostly of zeros. We also note recent work by \citet{coey2019improving} who considered experiment splitting as an alternative
to empirical Bayes estimation. Our paper adds to this body of knowledge by providing the
first characterization of minimax-optimal error in the general model~\eqref{eq:np-fay-herriot-model},
by proposing a flexible estimator that attains this bound up to constants, and
by studying robustness of non-parametric empirical Bayes methods to model misspecification.

% -- 

%Returning to the problem at hand, our general strategy towards optimal shrinkage estimation with covariates, amounts to simultaneously learning $A$ and $m(\cdot)$, so that we may implement~\eqref{eq:Bayes_rule}. However, there is a further catch: We would like the resulting procedure to provide guarantees,\emph{even} if~\eqref{eq:np-fay-herriot-model}, but instead the following considerably more general model holds:
%\begin{equation}
%\label{eq:np-fay-herriot-compound-model-variance}
%(X, \mu) \sim \mathbb{P}^{(X,\mu)},\;\; Z \mid \mu, X \sim \p{\mu, \sigma^2},\;
%\text{ i.e. }\EE{Z \mid \mu, X} = \mu, \; \Var{Z \mid \mu, X} = \sigma^2
%\end{equation}
%Guarantees under such a general setting are necessary if we want to apply the method as an off-the-shelf solution to denoising problems with covariates.

%Finally, we also note recent work on covariate-powered multiple testing \cite{ignatiadis2016data,lei2018adapt,li2016multiple} which, in an analogy to our results, enables considerable power gains over classical methods like the Benjamini-Hochberg procedure that ignore covariates.

\section{Minimax rates for empirical Bayes estimation with covariates}
\label{sec:minimax}
We first develop minimax optimality theory for model~\eqref{eq:np-fay-herriot-model}, when $m$ is known to lie in a class $\mathcal{C}$ of functions. To this end, we formalize the notion of regret in empirical Bayes estimation, following~\citet{robbins1964empirical}. Concretely, as before, we assume that we have access to $n$ i.i.d. copies $(X_i, Z_i)$ from model~\eqref{eq:np-fay-herriot-model}; $\mu_i$ is not observed. Our task at hand then is to construct a denoiser \smash{$\hat{t}_n:  \mathcal{X} \times \RR  \to \RR$} that we will use to estimate $\mu_{n+1}$ by \smash{$\hat{t}_n(X_{n+1}, Z_{n+1})$} for a future sample $(X_{n+1}, \, Z_{n+1})$. We benchmark this estimator against the unknown Bayes estimator \smash{$t^*_{m,A}(X_{n+1}, Z_{n+1})$} from ~\eqref{eq:Bayes_rule} in terms of its regret (excess risk) \smash{$L(\hat{t}_n; m, A)$}, where:
\begin{equation}
\label{eq:excess_risk}
L(t; m, A) := \EE[m, A]{\p{t(X_{n+1},Z_{n+1}) - \mu_{n+1}}^2} -  \EE[m,A]{\p{t_{m,A}^*(X_{n+1},Z_{n+1}) - \mu_{n+1}}^2}
\end{equation}
We characterize the difficulty of this task by exhibiting the minimax rates for the empirical Bayes excess risk incurred by not knowing $m \in \mathcal{C}$  (but knowing $A$), where $\mathcal{C}$ is a pre-specified class of functions:\footnote{We will propose procedures adaptive to unknown $A$ in Section~\ref{sec:sample_split_eb}.}
\begin{equation}
\label{eq:expected_eb_regret}
\mathfrak{M}_n^{\text{EB}}\p{\mathcal{C}; A, \sigma^2} := \inf_{\hat{t}_n} \sup_{m \in \mathcal{C}} \cb{\EE[m,A]{L(\hat{t}_n; m, A)}}
\end{equation}
Our key result, informally stated, is that the minimax excess
risk \smash{$\mathfrak{M}_n^{\text{EB}}$} can be characterized in terms of the minimax risk for estimating $m(\cdot)$ with
respect to $L^2(\mathbb P^X)$ in the regression problem in which we observe
$\smash{(X_i,Z_i)_{1 \leq i \leq n}}$ with $\smash{Z_i \mid X_i \sim \nn(m(X_i), \, A+\sigma^2)}$, i.e.,
\begin{equation}
\label{eq:nonparametric_minimax_rate}
\mathfrak{M}_n^{\text{Reg}}\p{\mathcal{C} ; A+\sigma^2} := \inf_{\hat{m}_n} \sup_{m \in \mathcal{C}}\EE[m, A]{\int \p{\hat{m}_n(x) - m(x)}^2 d\mathbb P^X(x)},
\end{equation}
such that, for many commonly used function classes $\mathcal{C}$, we have \footnote{Throughout,
we use the following notation for the asymptotic rates: For two sequences $a_n, b_n >0$, we say
$a_n \lesssim b_n$ if $\limsup_{n \to \infty} a_n/b_n \leq c$ for a constant $c$ that \emph{does not}
depend on $A, \sigma, n$. Similarly, we say $a_n \gtrsim b_n$ if $b_n \lesssim a_n$ and finally
$a_n \asymp b_n$ if both $a_n \gtrsim b_n$ and $a_n \lesssim b_n$.}
 \begin{equation}
 \label{eq:risk_scaling}
 \mathfrak{M}_n^{\text{EB}}\p{\mathcal{C}; A,\sigma^2} \asymp \frac{\sigma^4}{\p{\sigma^2 + A}^2}\mathfrak{M}_n^{\text{Reg}}\p{\mathcal{C}; A+\sigma^2}.
 \end{equation}
In other words, when $A/\sigma^2$ is very large, we find that it is easy to match the performance of Bayes rule~\eqref{eq:Bayes_rule},
since it collapses to $Z_i$. On the other hand, when $A/\sigma^2$ is small, matching the Bayes rule requires estimating $m(\cdot)$
well, and~\eqref{eq:risk_scaling} precisely describes how the difficulty of estimating $m(\cdot)$ affects our problem of interest.

Previous work on minimax rates for the excess risk~\eqref{eq:excess_risk} has been sparse; some exceptions include \citet{benhaddou2013adaptive}, \citet{li2005convergence} and \citet{penskaya1995lower}, who develop minimax bounds on~\eqref{eq:excess_risk} when $\mu \sim G, Z \mid \mu \sim \nn\p{0,\sigma^2}$, i.e., in the setting without covariates but with potentially more general priors. Beyond the modulation through covariates, a crucial difference of our approach is that we pay attention to the behavior in terms of $A$ and $\sigma$, instead of absorbing them into constants.

\paragraph{Lower bound} Here we provide a lemma for deriving lower bounds for worst case expected excess risk~\eqref{eq:expected_eb_regret} through reduction to hypothesis testing. The result is applicable to any class $\mathcal{C}$ for which we can prove a lower bound on the minimax regression error using Le Cam's two point method or Fano's method~\citep{duchi2019lecture, gyorfi2006distribution, ibragimov2013statistical, tsybakov2008introduction}; we will provide concrete examples below.
\begin{lemm}
\label{lemm:minimax}
For each $n$, let $\mathcal{V}_n$ be a finite set and $\mathcal{C}_n=\cb{m_{n,v} \mid v \in \mathcal{V}_n} \subset \mathcal{C}$ be a collection of functions indexed by $\mathcal{V}_n$ such that for a sequence $\delta_n > 0$:
$$ \int\p{m_{n,v}(x) - m_{n, v'}(x)}^2 d\mathbb P^X(x) \geq \delta_n^2 \;\text{ for all } v\neq v' \in \mathcal{V}_n,\; \text{for all}\; n$$
If furthermore, $\sup_{v,v' \in \mathcal{V}_n}\sup_{x}\p{m_{n,v}(x) - m_{n, v'}(x)}^2 \to 0 \text{ as } n \to \infty$, then:
$$ \mathfrak{M}_n^{\text{EB}}\p{\mathcal{C}; A,\sigma^2} \gtrsim  \frac{\sigma^4}{\p{\sigma^2 + A}^2} \cdot \delta^2_n \cdot \inf_{\hat{V}_n} \mathbb P[\hat{V}_n \neq V_n]$$
Here, \smash{$\inf_{\hat{V}_n}\mathbb P[\hat{V}_n \neq V_n]$} is to be interpreted as follows: $V_n$ is drawn uniformly from $\mathcal{V}_n$ and conditionally on $V_n=v$, we draw the pairs $(X_i,Z_i)_{1 \leq i \leq n}$ from model~\eqref{eq:np-fay-herriot-model} with regression function $m_{n,v}(\cdot)$. The infimum is taken over all estimators $\hat{V}_n$ that are measurable with respect to  $(X_i,Z_i)_{1 \leq i \leq n}$.
\end{lemm}
The Lemma may be interpreted as follows: If information theoretically we cannot determine which $m_{n,v} \in \mathcal{C}_n$ generated $(X_i,Z_i)_{1 \leq i \leq n}$, yet the $m_{n,v}$ are well separated in $L^2(\mathbb P^X)$ norm, then the minimax empirical Bayes regret~\eqref{eq:expected_eb_regret} must be large. Proving lower bounds involves contructing $\mathcal{C}_n$.
\paragraph{Upper bound} Previously, we described the relationship of model~\eqref{eq:np-fay-herriot-model} to nonparametric regression. However, there is a further connection: Under~\eqref{eq:np-fay-herriot-model}, it also holds that $Z_i \mid X_i \sim \nn\p{m(X_i), \sigma^2 + A}$. Thus $m(\cdot)$ may estimated from the data by directly running a regression $Z_i \sim X_i$. Then, for known $A$, the natural impetus to approximate~\eqref{eq:Bayes_rule} in a data-driven way is to use a plug-in estimator. Concretely, given a $\hat{m}_n$ that achieves the minimax risk~\eqref{eq:nonparametric_minimax_rate}, we just plug that into the Bayes rule~\eqref{eq:Bayes_rule}:
\begin{equation}
\label{eq:plugin_hat_t}
\hat{t}_n(x,z) := t^*_{\hat{m}_n,A}(x,z) = \frac{A}{\sigma^2 + A} z + \frac{\sigma^2}{\sigma^2 + A} \hat{m}_n(x)
\end{equation}
This plug-in estimator, establishes the following upper bound on~\eqref{eq:expected_eb_regret}:
\begin{theo}
\label{theo:upper_bd}
Under model~\eqref{eq:np-fay-herriot-model}, it holds that:
$$\mathfrak{M}_n^{\text{EB}}\p{\mathcal{C}; A,\sigma^2} \leq \frac{\sigma^4}{\p{\sigma^2 + A}^2}\mathfrak{M}_n^{\text{Reg}}\p{\mathcal{C}; A+\sigma^2}$$
\end{theo}
In deriving the lower bound Lemma~\eqref{lemm:minimax}, the estimators considered may use the unknown $A$. For this reason, for the upper bound we also benchmark against estimators that know $A$; however in Section~\ref{sec:sample_split_eb} we demonstrate that in fact knowledge of $A$ is not required to attain optimal rates. Next we provide two concrete examples of classes, where the lower and upper bounds match up to constants.

\paragraph{The linear class (Fay-Herriot shrinkage)} As a first, simple example, we consider the model of~\citet{fay1979estimates}, in which: $\mathcal{X} = \RR^d,\;\text{and}\; \mathcal{C} = \text{Lin}\p{\RR^d} = \cb{ m \mid m(x) = x^\top \beta,\; \beta \in \RR^d}$.

\begin{theo}
\label{theo:linreg}
Assume the $X_i$ are $\simiid \nn\p{0, \Sigma}$ for an unknown covariance matrix $\Sigma \succ 0, \Sigma \in \RR^{d \times d}$. Then there exists a constant $C_{\text{Lin}}$ (which does not depend on the problem parameters) such that:

$$ \lim_{n \to \infty} \abs{\log\p{\mathfrak{M}_n^{\text{EB}}\p{\text{Lin}\p{\RR^d}; A,\sigma^2} \bigg / \frac{\sigma^4}{\p{\sigma^2 + A}^2} \cdot \frac{(\sigma^2+A)d}{n}}} \leq C_{\text{Lin}}$$
\end{theo}

\paragraph{The Lipschitz class} Next we let $\mathcal{X}=[0,1]^d$ and for $L>0$ we consider the Lipschitz class:
$$\mathcal{C} = \text{Lip}([0,1]^d, L):= \cb{m: [0,1]^d \to \RR \mid  \abs{m(x) -m(x')} \leq L\Norm{x-x'}_2 \;\forall\; x,x' \in [0,1]^d}. $$

\begin{theo}
\label{theo:hoelder}
Assume the $X_i$ are $\simiid F^X$, where $F^X$ is a measure on $[0,1]^d$ with Lebesgue density $f^X$ that satisfies $ \eta \leq f^X(u) \leq 1/\eta$ for all $u \in [0,1]^d$ for some $\eta > 0$. Then there exists a constant $C_{\text{Lip}}(d, \eta)$ which depends only on $d, \eta$ such that:
$$ \lim_{n \to \infty} \abs{\log\p{\mathfrak{M}_n^{\text{EB}}\p{\text{Lip}([0,1]^d, L); A,\sigma^2} \bigg / \frac{\sigma^4}{\p{\sigma^2 + A}^2} \cdot \p{\frac{ L^d \p{\sigma^2+A}}{n}}^{\frac{2}{2 + d}}}} \leq C_{\text{Lip}}(d, \eta)$$
\end{theo}

\section{Feasible estimation via split-sample empirical Bayes}
\label{sec:sample_split_eb}

The minimax estimator in~\eqref{eq:plugin_hat_t} that implements~\eqref{eq:Bayes_rule} in a data-driven
way is not feasible, because $A$ is unknown in practice. In principle, $A+\sigma^2$ (with $\sigma^2$ known)
is just $\Var{ Z_i \mid X_i}$, hence deriving a plug-in estimator for $A$ just takes us to the realm
of variance estimation in regression problems. But variance estimation for the general setting we consider here is
a notoriously difficult problem, with only partial solutions available for very specific settings
\citep[e.g.,][]{janson2017eigenprism,reid2016study}. Furthermore, even for 1-dimensional smooth nonparametric
regression the minimax rates for variance estimation may be slower than parametric~\citep{brown2007variance, shen2019optimal}.

Fortunately, it turns out that we do not need to accurately estimate $A$ in~\eqref{eq:np-fay-herriot-model} in order for our
approach to perform well. Rather, as shown below, if we naively read off an estimate of $A$ derived via sample splitting
as in \eqref{eq:A_sure_homoskedastic}, we still obtain strong guarantees. Concretely, we study the following algorithm:
\begin{enumerate}
\item Form a partition of $\cb{1,\dotsc,n}$ into two folds $I_1$ and $I_2$.
\item Use observations in $I_1$, to estimate the regression $ m(x) = \EE{ Z_i \mid X_i=x}$ by $\hat{m}_{I_1}(\cdot)$.
\item Use observations in $I_2$, to estimate $A$, through the formula:
\begin{equation}
\label{eq:A_sure_homoskedastic}
\hat{A}_{I_2} = \p{\frac{1}{\abs{I_2}}\sum_{i \in I_2} \p{\hat{m}_{I_1}(X_{i})-Z_i}^2 - \sigma^2}_+
\end{equation}
\item The estimated denoiser is then $\hat{t}_n^{\text{EBCF}}(\cdot, \cdot) = t^*_{\hat{m}_{I_1}, \hat{A}_{I_2} }(\cdot, \cdot)$.
\end{enumerate}
We prove the following guarantee for this estimator. In particular, the following implies that if the minimax rate for regression \eqref{eq:nonparametric_minimax_rate} is slower than the parametric rate $1/n$
and if $\abs{I_1}/n$ converges to a non-trivial limit, then our algorithm attains the minimax rate even when $A$ is unknown.
\begin{theo}
\label{theo:optimal_A_datadriven}
Consider a split of the data into two folds $I_1, I_2$, where $n_1 = \abs{I_1}, n_2 = \abs{I_2}$. Furthermore assume that $\hat{m}_{I_1}$ satisfies
$\mathbb{E}_{m,A}[\hat{m}_{I_1}(X)^4 \mid \hat{m}_{I_1}] \leq M$
almost surely for some $M<\infty$, where $X$ is a fresh draw from $\mathbb P^X$. Then the estimator $\hat{t}_n^{\text{EBCF}}$ satisfies the following guarantee:
$$ \EE[m,A]{ L\p{\hat{t}_n^{\text{EBCF}};\; m, A}} \leq \EE[m,A]{L\p{t^*_{\hat{m}_{n_1},A};\; m, A}}  + \frac{1}{n_2} O\p{1}$$
\end{theo}
We emphasize that this result does not depend on \smash{$\hA$} from \eqref{eq:A_sure_homoskedastic}
being a particularly accurate estimate of $A$.
Rather, what's driving our result is the following fact:
If \eqref{eq:np-fay-herriot-model} holds, but we use~\eqref{eq:Bayes_rule} with \smash{$\tilde{m}(\cdot) \neq m(\cdot)$},
then the choice of \smash{$\tilde{A}$} that minimizes the Bayes risk among all estimators of the form
\smash{$\smash{t^*_{\tilde{m}, \tilde{A}}(\cdot,\cdot),\; \tilde{A} \geq 0}$} is not $A$, but rather (cf. derivation in Proposition~\ref{prop:best_A} of the Appendix)
$$\smash{A_{\tilde{m}} :=  \EE[m,A]{\p{\tilde{m}(X_{n+1}) - Z_{n+1}}^2} - \sigma^2 =  A + \EE[m,A]{\p{\tilde{m}(X_{n+1}) - m(X_{n+1})}^2}}.$$  
In other words, we're better off inflating the prior variance to account for the additional estimation error of \smash{$\tm(\cdot)$};
and this inflated prior variance is exactly what's captured in \eqref{eq:A_sure_homoskedastic}.

\section{Robustness to misspecification}
\label{sec:compound}
So far, our results and estimator apply to Robbins' model~\citep{robbins1964empirical} in
which~\eqref{eq:np-fay-herriot-model} holds and we are interested in a estimating a future $\mu_{n+1}$.
However, it is also of considerable interest to understand the behavior of empirical Bayes estimation
when the specification~\eqref{eq:np-fay-herriot-model} doesn't hold. In this section, we consider
properties of our estimator under the weaker assumption that we only have a generic data-generating
distribution for $(X_i,\, \mu_i, \, Z_i)$ of the form
\begin{equation}
\label{eq:np-fay-herriot-compound-model-variance}
(X_i, \mu_i) \sim \mathbb{P}^{(X_i,\mu_i)},\;\; 
\EE{Z_i \mid \mu_i, X_i} = \mu_i, \;\;
\Var{Z_i \mid \mu_i, X_i} = \sigma^2,
\end{equation}
and we seek to estimate the unknown $\mu_1, \, \dotsc, \, \mu_{n}$ underlying the
experiments we have data for. The distributions indexed by $i$ are assumed to be independent, but need not be identical. This setting is sometimes referred to as the compound
estimation problem \citep{brown2009nonparametric}.

We proceed with a cross-fold estimator, which we call EBCF (empirical Bayes with cross-fitting), as follows: We split the data as above, but now also consider flipping
the roles of $I_1$ and $I_2$ such that we can make predictions \smash{$\hmu_i$} for all $i = 1, \, ..., \, n$ as
$$ \hat{\mu}_i^{\text{EBCF}} = t^*_{\hat{m}_{I_1}, \hat{A}_{I_2} }(X_i, Z_i) \text{ for } i \in I_2 \;\;\;\& \;\;\; \hat{\mu}_i^{\text{EBCF}} = t^*_{\hat{m}_{I_2}, \hat{A}_{I_1} }(X_i, Z_i) \text{ for } i \in I_1. $$
This is a 2-fold cross-fitting scheme, which has been
fruitful in causal inference~\citep{chernozhukov2017double, nie2017learning, schick1986asymptotically} and
multiple testing~\citep{ignatiadis2016data, ignatiadis2018covariate}.
We also note that extensions to $k$-fold cross-fitting are immediate.

\paragraph{SURE for empirical Bayes}

The key property of our estimator that enables our approach to be robust outside of the strict model~\eqref{eq:np-fay-herriot-model}
is as follows. Let $\text{SURE}(\cdot)$ denote Stein's Unbiased Risk Estimate, a flexible risk estimator that is motivated by the
study of estimators for $\mu_i$ in the Gaussian model \smash{$Z_i \sim \nn(\mu_i, \, \sigma^2)$} \citep{stein1981estimation}.
Then, although our estimator was not originally motivated by SURE, one can algebraically verify that our estimator with
a plug-in choice of \smash{$\hA$} in fact minimizes SURE among all comparable shrinkage estimators
(the same holds true with $I_1, I_2$ flipped):
\begin{equation}
\label{eq:hat_A_sure}
\begin{aligned}
&\hat{A}_{I_2} = \p{\frac{1}{\abs{I_2}}\sum_{i \in I_2} \p{\hat{m}_{I_1}(X_{i})-Z_i}^2 - \sigma^2}_+ 
\Longleftrightarrow\;\; \hat{A}_{I_2} = \argminB_{A \geq 0}\cb{\SURE_{I_2}(A)},\\
&\text{where } \SURE_{I_2}(A) := \frac{1}{\abs{I_2}}\sum_{i \in I_2} \p{\sigma^2 + \frac{\sigma^4}{(A+\sigma^2)^2}(Z_i - \hat{m}_{I_1}(X_i))^2 - 2\frac{\sigma^4}{ A + \sigma^2}}.
\end{aligned}
\end{equation}
Furthermore, SURE has the following remarkable property in our setting:
For any data-generating process as in \eqref{eq:np-fay-herriot-compound-model-variance}
and any $A \geq 0$ \citep[see also][]{jiang2011best,kou2017optimal,xie2012sure},
\begin{equation}
\label{eq:sure_misspec}
\EE{\SURE_{I_2}(A) \mid X_{1:n}, \mu_{1:n}} = \frac{1}{\abs{I_2}}\sum_{i \in I_2} \EE{\p{\mu_i - t^*_{\hat{m}_{I_1}, A}(X_i, Z_i)}^2\mid X_{1:n},\mu_{1:n}},
\end{equation}
even when the distribution of $Z_i$ conditionally on $\mu_i$ and $X_i$ is not Gaussian.
Putting \eqref{eq:hat_A_sure} and \eqref{eq:sure_misspec} together, we find that we can argue using SURE
that our estimator minimizes an unbiased risk estimate for the generic specification
\eqref{eq:np-fay-herriot-compound-model-variance}, despite the fact that our procedure was
not directly motivated by SURE and SURE itself was only designed for Gaussian estimation.

\paragraph{Gaussian data with equal variance and James-Stein property} To derive a first consequence of the above, let us first focus on a special case of \eqref{eq:np-fay-herriot-compound-model-variance}, where $Z_i \mid (\mu_i, X_i) \sim \nn\p{\mu_i, \sigma^2}$. Then the EBCF estimate satisfies the James-Stein property of strictly dominating the direct estimator $Z_i$~\citep{james1961estimation}\footnote{\citet{li1984data} provide a similar result when $\hat{m}(\cdot)$ is a linear smoother.}. In other words, even if one has covariates $X_i$, which are uninformative, or one uses a really poor method for prediction, one still does no worse than just using \smash{$\hmu_i :=Z_i$}.
\begin{theo}[James-Stein property]
\label{theo:james_stein} Under the assumptions above and if $\abs{I_1}, \abs{I_2} \geq 5$, the proposed estimator $\hat{\mu}_i$ uniformly dominates the (conditional) maximum likelihood estimator $Z_i$, in other words for all $\mu_1,\dotsc,\mu_n$ and $X_1, \dotsc, X_n$, it holds that:
$$ \frac{1}{n}\sum_{i=1}^n \EE{ (\mu_i - \hat{\mu}_i^{\text{EBCF}})^2 \mid X_{1:n}, \mathbf{\mu}_{1:n}} <  \frac{1}{n}\sum_{i=1}^n \EE{ (\mu_i - Z_i)^2 \mid X_{1:n}, \mathbf{\mu}_{1:n}} = \sigma^2 $$
\end{theo}
\paragraph{Non-Gaussian data with equal variance}
Next we drop the Gaussianity assumption, and consider the model~\eqref{eq:np-fay-herriot-compound-model-variance}
in full generality. We use properties of SURE outlined above to establish the following:
\begin{theo}
\label{theo:sure_equal_variance}
Assume the pairs $(X_i, Z_i)_{1 \leq i \leq n}$ are independent and satisfy~\eqref{eq:np-fay-herriot-compound-model-variance}. Furthermore assume that there exist $\Gamma, M < \infty$ such that $\sup_{i} \EE{Z_i^4 \mid \mu_i, X_i} \leq \Gamma^4$ and that $\sup_i \abs{\mu_i} \leq M$, $\sup_{x} \abs{\hat{m}_{I_1}(x)} \leq M$ almost surely. Then (the analogous claim holds also with $I_1, I_2$ flipped):
{\small
$$
\begin{aligned}
\sup_{A \geq 0}\cb{ \frac{1}{\abs{I_2}}\sum_{i \in I_2} \EE{ \p{\mu_i - \hat{\mu}_i^{\text{EBCF}}}^2  -  \p{\mu_i - t^*_{\hat{m}_{I_{1}},A}(X_i, Z_i)}^2 \cond X_{1:n}, \mathbf{\mu}_{1:n}, Z_{I_{1}}}  } \leq  O\p{ \frac{1}{\sqrt{\abs{I_2}}}} 
\end{aligned}
$$}
\end{theo}

\begin{coro}
\label{coro:equal_sample_size_compound}
Assume that $\abs{I_1}=\abs{I_2}=n/2$ and $(X_i, \mu_i, Z_i)$ are i.i.d. and satisfy the assumptions of Theorem~\ref{theo:sure_equal_variance}.
Then, the following holds, with $(X, \, \mu)$ a fresh draw from \eqref{eq:np-fay-herriot-compound-model-variance}:
\begin{equation}
\label{eq:equal_sample_size_compound}
 \frac{1}{n}\sum_{i=1}^n \EE{ \p{\mu_i - \hat{\mu}_i^{\text{EBCF}}}^2 } \leq \frac{\sigma^2 \EE{\p{\hat{m}_{n/2}(X) - \mu}^2}}{ \sigma^2 + \EE{\p{\hat{m}_{n/2}(X) - \mu}^2}} + O\p{ \frac{1}{\sqrt{n}}}.
\end{equation}
\end{coro}
Here \smash{$\hat{m}_{n/2}(\cdot)$} is the fitted function based on $n/2$ samples $(X_i, Z_i)$.
To interpret this result, we note that when \smash{$\hat{m}(\cdot)$} can accurately capture $\mu_i$,
i.e., \smash{$\hat{m}(\cdot)$} is a good estimate of $m(\cdot)$ and $\mu_i$ can be well explained using the
available covariates $X_i$, the error in \eqref{eq:equal_sample_size_compound} essentially matches the error of
the direct regression-based method \smash{$\hmu_i := \hat{m}_{n/2}(X_i)$}. Conversely, when the error of
\smash{$\hat{m}(\cdot)$} for estimating \smash{$\mu_i$} is large, we recover the error $\sigma^2$ of the
simple estimator \smash{$\hmu_i := Z_i$}. But in the interesting regime where the mean-squared error of
\smash{$\hat{m}(\cdot)$} for \smash{$\mu_i$} is comparable to $\sigma^2$, we can do a much better job by
taking a convex combination of the regression prediction \smash{$\hat{m}_{n/2}(X_i)$} and $Z_i$, 
and the EBCF estimator automatically and robustly navigates this trade-off.

\paragraph{Non-Gaussian data with unequal variance:} Finally, we note that we may even drop the assumption of equal variance and assume each unit has its own (conditional) variance $\sigma_i^2$ in ~\eqref{eq:np-fay-herriot-compound-model-variance} rather than the same $\sigma^2$ for everyone. We may think of the Bayes estimator~\eqref{eq:Bayes_rule} as also being a function of $\sigma_i$, i.e. write it as \smash{$t^*_{m,A}(x,z, \sigma)$}. Then, the EBCF estimator takes the following form: For $i \in I_2$ we estimate $\mu_i$ by \smash{$t^*_{\hat{m}_{I_1}, \hat{A}_{I_2}}(X_i, Z_i, \sigma_i)$}. We get \smash{$\hat{m}_{I_1}$} by regression, while for \smash{$\hat{A}_{I_2}$}, we generalize~\eqref{eq:hat_A_sure}:
{\footnotesize
\begin{equation*}
\label{eq:A_sure_heteroskedastic}
\begin{aligned}
& \hat{A}_{I_2} = \argminB_{A \geq 0}\cb{\SURE_{I_2}(A)},\; \SURE_{I_2}(A) = \frac{1}{\abs{I_2}}\sum_{i \in I_2} \p{\sigma_i^2 + \frac{\sigma_i^4}{(A+\sigma_i^2)^2}(Z_i - \hat{m}_{I_1}(X_i))^2 - 2\frac{\sigma_i^4}{ A + \sigma_i^2}}
\end{aligned}
\end{equation*}}
The result of Theorem~\ref{theo:sure_equal_variance} (see Appendix~\ref{sec:SURE_results}) also holds in this case and we demonstrate the claims in the empirical application on the MovieLens dataset below.

\section{Empirical results}
For our empirical results we compare the following 4 estimation methods for $\mu_i$:
\textbf{a)} The unbiased estimator \smash{$\hmu_i := Z_i$},
\textbf{b)} the out-of-fold \footnote{By out-of-fold we mean that the regression prediction is the one used by 5-fold EBCF described below.} regression prediction \smash{$\hmu_i := \hat{m}(X_i)$}, where $\hat{m}$ is the fit from boosted regression trees, as implemented in XGBoost~\citep{chen2016xgboost} with number of iterations chosen by $5$-fold cross-validation and $\eta =0.1$ (weight with which new trees are added to the ensemble),
\textbf{c)} the empirical Bayes estimator~\eqref{eq:Bayes_rule} without covariates that shrinks $Z_i$ towards the grand average $\sum_{i=1}^n Z_i/n$,
with tuning parameters selected via SURE following~\citep{xie2012sure},
and \textbf{d)} the proposed EBCF (empirical Bayes with cross-fitting) method, with 5 folds used for cross-fitting and XGBoost as the regression learner (with cross-validation nested within cross-fitting).

\begin{figure}
    \centering
    \includegraphics[width=0.9\linewidth]{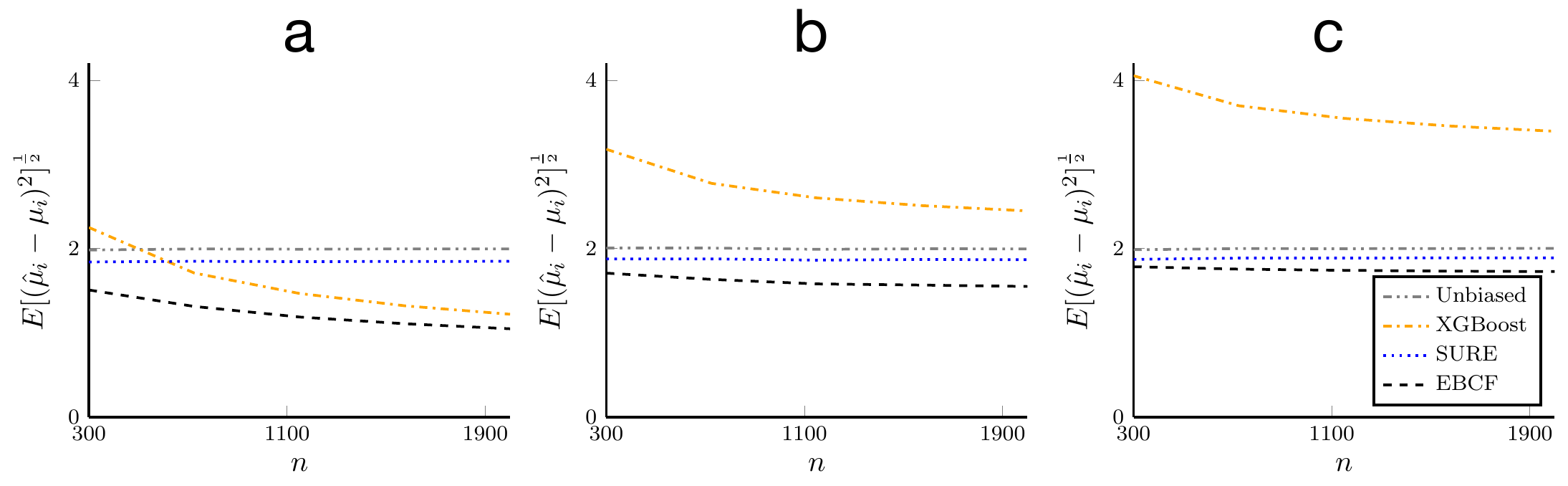}
    \caption{\textbf{Root mean squared error (RMSE) for estimating $\mu_i$ in model~\eqref{eq:np-fay-herriot-model}}. Results are shown as a function of $n$ for the four estimators described in the main text. \textbf{a)} Here we let $\sigma=2, A=0$ corresponding to the case of nonparametric regression.  In panel \textbf{b)}, we let $\sigma=\sqrt{A}=2.0$ corresponding to intermediate shrinkage and in panel \textbf{c)} we let $\sigma=2, \sqrt{A}=3$. The standard errors of all RMSEs are smaller or equal to 0.01.}
    \label{fig:simulations}
\end{figure}

\textbf{Synthetic data}: We generate data from model~\eqref{eq:np-fay-herriot-model} with $\mathbb P^X = U[0,1]^{15}$ and $m(\cdot)$ is the~\citet{friedman1991multivariate} function  $m(x)=10\sin(\pi x_1x_2) + 20(x_3 - 1/2)^2 + 10x_4 +5x_5$, and the last 10 coordinates are noise. Furthermore, we let $\sigma=2.0$ and vary $A \in \cb{0, 4, 9}$, mimicking the three cases in Figure~\ref{fig:np_vs_fh}, and we also vary $n$. Results are averaged over 100 simulations and shown in Figure~\ref{fig:simulations}. We make the following observation: The unbiased estimator $Z_i$ and the SURE estimator which shrinks towards the grand mean have constant mean squared error and results do not improve with increasing $n$. The XGBoost predictor improves with increasing $n$, since $m(\cdot)$ is estimated more accurately; indeed in panel a), if $\hat{m}(\cdot)$ would be exactly equal to $m(\cdot)$, then the error would be $0$. However, as seen in panels $b,c)$, when $A>0$, the mean squared error of XGBoost is lower bounded by $A$, even under perfect prediction of $m(\cdot)$. In contrast, EBCF always improves with $n$ by leveraging the improved predictions of XGBoost, and outperforms all other estimators, even in the case $A=0$ which corresponds to nonparametric regression.

\begin{figure}
    \centering
    \includegraphics[width=0.95\linewidth]{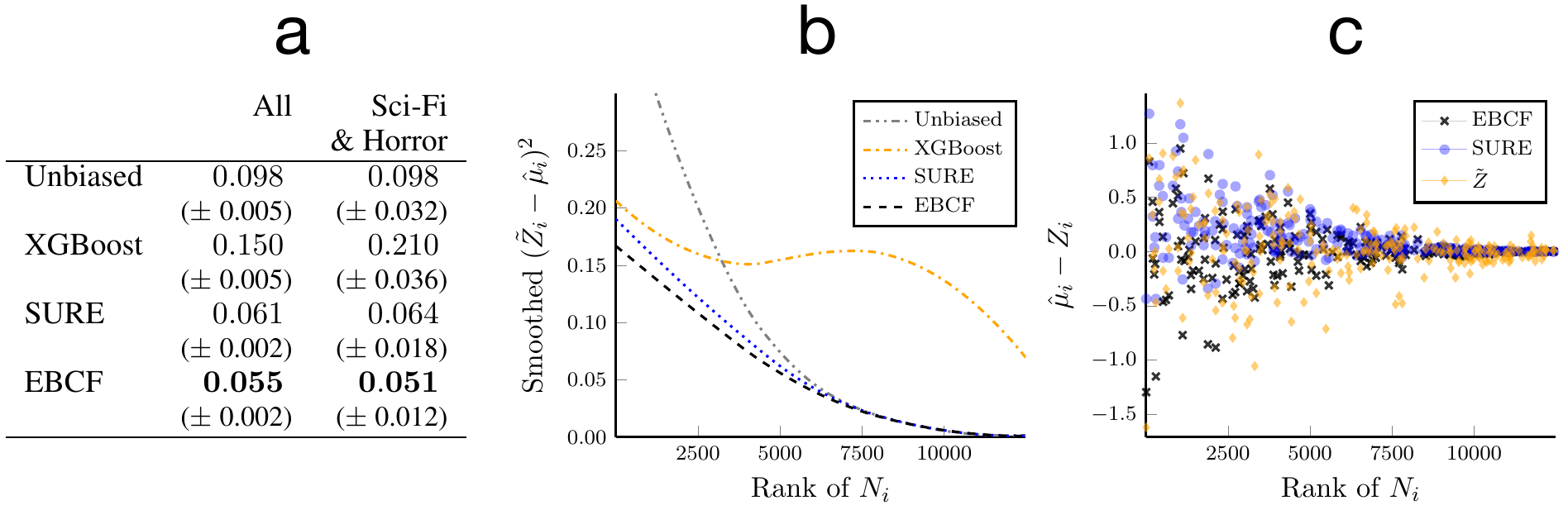}
    \caption{\textbf{EB analysis of the Movielens dataset for prediction of average movie rating. a)} Mean-squared error (MSE) $\smash{n^{-1} \sum_{i = 1}^n (\hat{\mu}_i-\tilde{Z}_i)^2}$ $(\pm\;2 \text{ standard errors of the MSE })$ of four estimators for the Movielens dataset (where $\tilde{Z}_i$ is the average rating computed from the heldout data with 90\% of users) for all movies, as well as the subset of movies that are classified as both Horror and Sci-Fi. \textbf{b)} LOESS smooth of mean squared error across all movies against the rank of $N_i$, where $N_i$ is the number of users that rated movie $i$ in the training set. \textbf{c)} Deviations of EBCF (empirical Bayes with cross-fitting) and SURE (Stein's unbiased risk estimate) predictions from the unbiased estimator $Z_i$ as a function of $N_i$ for all Horror \& Sci-Fi movies. We also show the ``true'' errors $\tilde{Z}_i - Z_i$.}
    \label{fig:movielens}
\end{figure}

\textbf{MovieLens data}~\citep{harper2016movielens}: Here we elaborate on the example from the introduction which aims to predict the average movie rating given ratings from a finite number of users. The MovieLens dataset consists of approximately 20 million ratings in $\cb{0, 0.5, \dotsc, 5}$ from 138,000 users applied to 27,000 movies. To demonstrate the applicability of our approach, when model~\eqref{eq:np-fay-herriot-model} does not necessarily hold, we randomly choose 10\% of all users and attempt to estimate the movie ratings from them. This corresponds to having a much smaller dataset. We then summarize the $i$-th movie, by $Z_i$, the average of the $N_i$ users (in the training dataset) that rated it. We further have covariates $X_i \in \RR^{20}$ that include $N_i$, the year the movie was released, as well as indicators of 18 genres to which the movie may belong (action, comedy, etc.). We posit that $Z_i \mid \mu_i, X_i \sim (\mu_i, \sigma^2/N_i)$ and want to estimate $\mu_i$.\footnote{We replace $\sigma^2$ by $\hat{\sigma}^2 \doteq 0.94$, the average of the sample standard deviations across all movies.}
As our pseudo ground truth for movie $i$ we use \smash{$\tilde{Z}_i$}, the average movie rating among the remaining $90\%$ of users and then report the error \smash{$\sum_{i=1}^n (\tilde{Z}_i - \hat{\mu}_i)^2/n$}, where $n$ is the total number of movies.\footnote{We filter movies and keep only movies with at least 3 ratings in the training set and 11 in the validation set.}

The average error across all movies is shown in Figure~\ref{fig:movielens}a; here the XGBoost predictor performs worst, followed by the unbiased estimator $Z_i$. Instead, the two EB approaches perform a lot better with EBCF scoring the lowest error. The same is true when comparing only the 253 movies with genre tags for both horror and Sci-Fi.
In panel b), we show the relationship between the error \smash{$(\tilde{Z}_i - \hat{\mu}_i)^2$}  and the rank of the per-movie
number of reviews $N_i$ using a LOESS smoother~\citep{cleveland1988locally}. We observe that the 3 estimators that
use $Z_i$, do a perfect job for large $N_i$ and a worse job for smaller $N_i$. In particular, the error of $Z_i$ blows up
at small $N_i$, and the error gains of EBCF occur precisely at low sample sizes. On the other hand, the XGBoost prediction
has an error that does not get reduced by larger $N$, but is competitive at small $N$. Panel c) shows \smash{$\hat{\mu}_i - Z_i$}
for the 253 predictions of EBCF and SURE for horror/Sci-Fi movies as a function of the rank of $N_i$. For large $N_i$, again
both EB estimators agree with the unbiased estimator. However, for small $N_i$, it appears that most Sci-Fi/Horror
movies are worse than the average movie, and EB without covariates tries to correct for this by assigning them a higher rating. Conversely,
EBCF automatically realizes that these movies tend to get low ratings, and pulls the unbiased estimator $Z_i$ further down.

 % \begin{tabular}{lrr}
 % & All & Sci-Fi \\
 % &  & \& Horror \\
%\hline
%Unbiased  & $0.098\;$  & $0.098\;$ \\
%& \footnotesize{($\pm$ 0.005)} &  \footnotesize{($\pm$ 0.032)} \\
%XGBoost & $0.150\;$  & $0.210\;$\\
%&  \footnotesize{($\pm$ 0.005)}  &  \footnotesize{($\pm$ 0.036)}\\
%SURE & $0.061\;$  & $0.064\;$ \\
%%& \footnotesize{($\pm$ 0.002)} & \footnotesize{($\pm$ 0.018)} \\
%EBCF   &   $\mathbf{0.055}\;$  &  $\mathbf{0.051\;}$ \\
%& \footnotesize{($\pm$ 0.002)} & \footnotesize{($\pm$ 0.012)}\\ \hline
%\end{tabular}

\textbf{Communities and Crimes data} from the UCI repository~\citep{Dua2019, redmond2002data}: The dataset provides information about the number of crimes in multiple US communities as compiled by the FBI Uniform Crime Reporting program in 1995. Our task is to predict the non-violent crime rate $p_i$ of community $i$, defined as $p_i :=  \text{Crimes in community }i/\text{Population } i$, for each of  $n=2118$ communities\footnote{We filter out communities with a missing number of non-violent crimes.}. We  observe a dataset in which the population of each community is down-sampled to $B=200$ as
$$ C_i \sim \text{Hypergeometric}(B, \;\text{Crimes in community }i,\; \text{Population } i)$$
We seek to predict $p_i$ based on $C_i$ and covariates $X_i \in \RR^{74}$ which include all unnormalized, numeric predictive covariates in the UCI data set description (after removing covariates with missing entries) and comprise features derived from Census and law enforcement data, such as percentage of people that are employed and percentage of police officers assigned to drug units. We note that the hypergeometric subsampling makes the estimation task harder and also provides us with pseudo ground truth $p_i$; cf.~\citet{wager2015efficiency} for further motivation of such down-sampling.

\begin{SCtable}
  \begin{tabular}{lrr}
  & $B=200$ &$B=500$ \\
  & MSE ($\times 10^6$)  & MSE ($\times 10^6$) \\
\hline
Unbiased  & $223.9$ \footnotesize{($\pm 16.8$)} & $92.2$ \footnotesize{($\pm 7.1$)}\\
XGBoost & $398.0$  \footnotesize{($\pm 81.8$)} & $370.2$ \footnotesize{($\pm 78.6$)}\\
SURE & $184.2$ \footnotesize{($\pm 18.9$)} & $85.6$ \footnotesize{($\pm 7.2$)}\\
EBCF   &  $\mathbf{152.0}$ \footnotesize{($\pm 22.2$)} & $\mathbf{78.5}$ \footnotesize{($\pm 10.3$)}\\ \hline
\end{tabular}
 \caption{\textbf{EB analysis of the Communities and Crimes dataset.} The table reports the mean-squared error $(\pm\;2 \text{ standard errors})$ of four different estimators for the non-violent crime rate. The columns correspond to down-sampling the dataset to a population of ${B=200}$ or $B=500$ for each community.}
\label{tab:communities_and_crimes}
\end{SCtable}

The problem may be cast into the setting of this paper by defining  \smash{$Z_i := \sqrt{C_i/B}$}. Then, by a variance stabilizing argument, it follows that \smash{$Z_i \mathrel{\dot\sim} \p{\sqrt{p_i}, 1/(4\cdot B)}$} and we may apply the same methods as in the preceding examples to estimate \smash{$\mu_i := \sqrt{p}_i$} by \smash{$\hat{\mu}_i$}. After transforming the estimates back to the original scale through $\hat{p}_i = \hat{\mu}_i^2$, we report the error \smash{$\sum_{i=1}^n (p_i - \hat{p}_i)^2/n$}, where $n$ is the number of communities analyzed.
Table~\ref{tab:communities_and_crimes} shows the results of this analysis, as well as the same analysis repeated for $B=500$. EBCF shows promising performance compared to the other baselines for both $B$. As we decrease the amount of downsampling from $B=200$ to $B=500$, we see that methods that depend on $Z_i$ (unbiased, SURE and EBCF) improve a lot, while XGBoost does not.

\section{Discussion}

Empirical Bayes is a powerful framework for pooling information across many experiments,
and improve the precision of our inference about each experiment on its own \citep{efron2010large,robbins1964empirical}.
Existing empirical Bayes methods, however, do not allow the analyst to leverage covariate information
unless they assume a rigid parametric model as in \citet{fay1979estimates}, or are willing to commit to a specific
end-to-end estimation strategy as in, e.g., \citet{opsomer2008non}. In contrast, the approach proposed here allows
the analyst to perform covariate-powered empirical Bayes estimation on the basis of any black-box predictive model,
and has strong formal properties whether or not the model \eqref{eq:np-fay-herriot-model} used to motivate our procedure
is well specified. Our approach may be extended in future work by considering generalizations of~\eqref{eq:np-fay-herriot-model}, such as covariate-based modulation of the prior variance, i.e., \smash{$\mu_i \cond X_i \sim \mathcal{N}(m(X_i),A(X_i))$}. The working assumption of a normal prior could also be replaced by heavy-tailed priors~\citep*{zhu2018} or priors with a point mass at zero.

The prevalence of settings where we need to analyze results from many loosely
related experiments seems only destined to grow, and we believe that empirical Bayes methods that allow for various
forms of structured side information hold promise for fruitful application across several different areas.

\vspace{0.2cm}
\subsubsection*{Code availability and reproducibility}
\vspace{-0.3cm}
The proposed EBCF (empirical Bayes with cross-fitting) method has been implemented in EBayes.jl (\url{https://github.com/nignatiadis/EBayes.jl}), a package written in the Julia language~\citep{bezanson2017julia}. Dependencies of EBayes.jl include MLJ.jl~\citep{MLJ}, Optim.jl~\citep{mogensen2018optim} and Distributions.jl~\citep{ besanccon2019distributions}. We also provide a Github repository (\url{https://github.com/nignatiadis/EBCrossFitPaper}) with code to reproduce all empirical results in this paper, including a specification for downloading the dependencies and datasets.

\newpage
\subsubsection*{Acknowledgments}
\vspace{-0.3cm}
The authors are grateful for enlightening conversations with Brad Efron,
Guido Imbens, Panagiotis Lolas and Paris Syminelakis. This research was funded by a gift from Google.

\bibliographystyle{plainnat}
\bibliography{covariate_ebayes}

\newpage
\begin{appendix}
\section{Proofs for Section~\ref{sec:minimax}}

\subsection{Proof of Theorem~\ref{theo:upper_bd}}

\begin{proof}

We will first show, that under model~\eqref{eq:np-fay-herriot-model}, the plug-in estimator~\eqref{eq:plugin_hat_t} satisfies:
\begin{equation}
\label{eq:plugin_excess_risk}
\EE[m,A]{L(\hat{t}_n; m, A)} = \frac{\sigma^4}{\p{\sigma^2+A}^2}\EE[m,A]{\p{\hat{m}_n(X_{n+1})-m(X_{n+1})}^2} 
\end{equation}
This also establishes the upper bound on the minimax excess risk if $\hat{m}_n$ is chosen in a minimax rate-optimal way for the regression problem. 

To prove~\eqref{eq:plugin_excess_risk}, we study the excess risk of this estimator conditionally on the covariate $X_{n+1}$ of the $n+1$-th observation:

{\small
\begin{equation*}
\begin{aligned}
 &\EE[m,A]{\p{\hat{t}(X_{n+1},Z_{n+1}) - \mu_{n+1}}^2 \cond X_{n+1}=x}  \\
 =\;\; &\EE[m,A]{\p{\frac{A}{\sigma^2 + A}Z_{n+1} +  \frac{\sigma^2}{\sigma^2 + A}\hat{m}(X_{n+1}) - \mu_{n+1} }^2 \cond X_{n+1}=x}\\
 =\;\;  &\EE[m,A]{\p{\frac{A}{\sigma^2 + A}Z_{n+1} +  \frac{\sigma^2}{\sigma^2 + A}m(X_{n+1}) - \mu_{n+1} +  \frac{\sigma^2}{\sigma^2 + A}\p{\hat{m}(X_{n+1})-m(X_{n+1})}}^2 \cond X_{n+1}=x}\\
 =\;\; &\EE[m,A]{\p{t^*_{m,A}(X_{n+1},Z_{n+1}) - \mu_{n+1}}^2 \cond X_{n+1}=x} + \frac{\sigma^4}{\p{\sigma^2+A}^2}\EE[m,A]{\p{\hat{m}(X_{n+1})-m(X_{n+1})}^2 \cond X_{n+1}=x}
\end{aligned}
\end{equation*}}

The result follows by integrating over $X_{n+1}$ and rearranging. 
\end{proof}

\subsection{Proof of Lemma~\ref{lemm:minimax}}

The idea of the proof follows the general paradigm in derivation of minimax optimal rates ~\citep{tsybakov2008introduction, duchi2019lecture} in which we reduce the original problem to a multiple hypothesis testing problem. More concretely, let us fix two functions $m_1, m_2 \in \mathcal{C}$ and call the induced distributions $P_1$,$P_2$. Say we have a denoiser $t(x,z)$ that performs extremely well under $m_1$ with respect to the loss~\eqref{eq:excess_risk}. Then we will argue that it cannot do too well under $m_2$. But then, given data $(X_1, Z_1), \dotsc, (X_n, Z_n)$ we may use the data-driven $\hat{t}(x,z)$ as a proxy for a hypothesis test: If its risk is small under $m_1$, but large under $m_2$, we would guess that $m_1$ is true and vice versa. Thus our task reduces to lower bounding the performance of a hypothesis test. These ideas will be made concrete in the arguments that follow.

Our proof strategy begins by studying the pointwise excess risk:

\begin{equation}
    \label{eq:excess_loss_pointwise}
    L(t; m, A \mid x) := \EE[m,A]{\p{t(x,Z_{n+1}) - \mu_{n+1}}^2 -  \p{t_{m,A}^*(x,Z_{n+1}) - \mu_{n+1}}^2 \mid X_{n+1}=x}
\end{equation}

\begin{lemm}
\label{lemm:bayes_risk_mixture}
There exist universal constants $c>0, \Delta>0$ such that when $\abs{m_1(x)-m_2(x)}/\sqrt{\sigma^2 +A} \leq \Delta$ (where $x$ is fixed, yet arbitrary) it holds for all $t$ that:
$$\frac{1}{2}\sqb{ L(t; m_1, A \mid x) + L(t; m_2, A \mid x)} \geq  c \frac{\sigma^4}{\p{\sigma^2 + A}^2}\p{m_1(x)-m_2(x)}^2$$
\end{lemm}

\begin{proof}
As a thought experiment, we consider the following generative model:

$$
\begin{aligned}
&\mu_{n+1} \; \sim \; G_x = \frac{1}{2}\sqb{\nn\p{m_1(x), A} + \nn\p{m_2(x), A}} \\
&Z_{n+1} \mid \mu_{n+1} \sim \nn\p{\mu_{n+1}, \sigma^2}
\end{aligned}
$$
Next consider the Bayes estimator for $\mu_{n+1}$ under this prior, namely:
\begin{equation}
t^*_{G_x}(z) := \EE[G_x]{ \mu_{n+1} \mid Z_{n+1}=z}
\end{equation}
Then, by definition of the Bayes estimator, it must hold that for any $t: \mathcal{X} \times \RR \to \RR$:
$$\EE[G_x]{\p{t(x, Z_{n+1}) - \mu_{n+1}}^2} \geq \EE[G_x]{\p{t^*_{G_x}(Z_{n+1}) - \mu_{n+1}}^2}$$
In the preceding result we are really thinking of $t$ as the curried function $t(x, \cdot)$. Next, by definition of $G_x$, the LHS of the above expression is the same as:
$$\frac{1}{2}\cb{\EE[m_1, A]{\p{t(x,Z_{n+1}) - \mu_{n+1}}^2 \mid X_{n+1}=x} + \EE[m_2, A]{\p{t(x,Z_{n+1}) - \mu_{n+1}}^2 \mid X_{n+1}=x} }$$
Also observe that $\inf_{t}\cb{\EE[m_1, A]{\p{t(x,Z_{n+1}) - \mu_{n+1}}^2 \mid X_{n+1}=x}} = A\sigma^2/(A+\sigma^2)$ and similarly for $m_2$, hence upon subtracting $A\sigma^2/(A+\sigma^2)$ from the above expression and its preceding inequality, we get:
$$ \frac{1}{2}\cb{ L(t; m_1, A \mid x) + L(t; m_2, A \mid x)} \geq  \EE[G_x]{\p{t^*_{G_x}(Z_{n+1}) - \mu_{n+1}}^2} - \frac{A\sigma^2}{A+\sigma^2}$$
Hence to conclude we will need to show that there exist universal constants $c, \Delta >0$ so that if $\abs{m_1(x)-m_2(x)}/\sqrt{\sigma^2 +A} \leq \Delta$:
\begin{equation}
\label{eq:univariate_bayes_regret}
\EE[G_x]{\p{t^*_{G_x}(Z_{n+1}) - \mu_{n+1}}^2} - \frac{A\sigma^2}{A+\sigma^2} \; \geq \; c \frac{\sigma^4}{\p{\sigma^2 + A}^2}\p{m_1(x)-m_2(x)}^2
\end{equation}
Note that the LHS depends on $m_1(x), m_2(x)$ through the definition of $G_x$. We provide the calculations and complete the proof in Appendix~\ref{subsec:to_testing_proof}.
\end{proof}

\begin{lemm}
\label{lemm:two_point_bound}
Let $c>0$, $\Delta >0$ the constants from Lemma~\ref{lemm:bayes_risk_mixture}. Then, for all $m_1, m_2:\mathcal{X} \to \RR$, the following implication holds for any $t: \mathcal{X} \times \RR \to \RR$ 

\begin{equation}
\begin{aligned}
&L(t; m_1, A) < c \frac{\sigma^4}{\p{\sigma^2 + A}^2}\int \p{m_1(x)-m_2(x)}^2  \ind\cb{\frac{\p{m_1(x)-m_2(x)}^2}{\sigma^2 + A}\leq \Delta^2} d\mathbb P^X(x) \\ \Longrightarrow \;\;\; &L(t; m_2, A) \geq c \frac{\sigma^4}{\p{\sigma^2 + A}^2}\int \p{m_1(x)-m_2(x)}^2  \ind\cb{\frac{\p{m_1(x)-m_2(x)}^2}{\sigma^2 + A}\leq \Delta^2}  d\mathbb P^X(x)
\end{aligned}
\end{equation}
\end{lemm}

\begin{proof}
We use the result from Lemma~\eqref{lemm:bayes_risk_mixture}, noting that $L(t; m, A) = \int L(t; m, A \mid x)d\mathbb P^X(x)$.

$$
\begin{aligned}
\frac{1}{2}\sqb{L(t; m_1, A) + L(t;m_2,A)} &= \int \frac{1}{2}\sqb{L(t; m_1, A \mid x) + L(t; m_2, A \mid x)}d\mathbb P^X(x) \\
&\geq \int \frac{1}{2}\sqb{L(t; m_1, A \mid x) + L(t; m_2, A \mid x)}  \ind\cb{\frac{\p{m_1(x)-m_2(x)}^2}{\sigma^2 + A} \leq \Delta^2} d\mathbb P^X(x) \\
& \geq c \frac{\sigma^4}{\p{\sigma^2 + A}^2}\int \p{m_1(x)-m_2(x)}^2 \ind\cb{\frac{\p{m_1(x)-m_2(x)}^2}{\sigma^2 + A}\leq \Delta^2}  d\mathbb P^X(x)
\end{aligned}
$$

Thus not both $L(t; m_1, A), L(t;m_2,A)$ may be $<$ than the RHS at the same time.

\end{proof}

The above lemma allows us to prove lower bounds by reduction to hypothesis testing. In particular, let us recall the statement from Lemma~\ref{lemm:minimax}, now stated in slightly more generality and dropping explicit notation for $n$ in the constructed collection of functions $\cb{m_v \mid v \in \mathcal{V}_n}$:

  \newtheoremstyle{TheoremNum}
        {\topsep}{\topsep}              %%% space between body and thm
        {\itshape}                      %%% Thm body font
        {}                              %%% Indent amount (empty = no indent)
        {\bfseries}                     %%% Thm head font
        {.}                             %%% Punctuation after thm head
        { }                             %%% Space after thm head
        {\thmname{#1}\thmnote{ \bfseries #3}}%%% Thm head spec
    \theoremstyle{TheoremNum}
    \newtheorem{thmn}{Lemma}

\begin{thmn}[~\ref{lemm:minimax} (More general version)]
For each $n$, let $\mathcal{V}_n$ be a finite set and $\cb{m_v \mid v \in \mathcal{V}_n} \subset \mathcal{C}$ be a collection of functions indexed by $\mathcal{V}_n$ such that for a sequence $\delta_n > 0$:

$$ \delta_n^2 \leq \int\p{m_v(x) - m_{v'}(x)}^2\ind\cb{\frac{\p{m_v(x)-m_{v'}(x)}^2}{\sigma^2 + A}\leq \Delta^2} d\mathbb P^X(x)  \;\text{ for all } v\neq v' \in \mathcal{V}_n, \forall n$$
Then:
$$ \mathfrak{M}_n^{\text{EB}}\p{\mathcal{C}; A,\sigma^2} \gtrsim  \frac{\sigma^4}{\p{\sigma^2 + A}^2} \cdot \delta^2_n \cdot \inf_{\hat{V}_n} \PP{\hat{V}_n \neq V_n}$$
Here, \smash{$\inf_{\hat{V}_n}\mathbb P[\hat{V}_n \neq V_n]$} is to be interpreted as follows: $V_n$ is drawn uniformly from $\mathcal{V}_n$ and conditionally on $V_n=v$, we draw the pairs $(X_i,Z_i)_{1 \leq i \leq n}$ from model~\eqref{eq:np-fay-herriot-model} with regression function $m_{n,v}(\cdot)$. The infimum is taken over all estimators $\hat{V}_n$ that are measurable with respect to  $(X_i,Z_i)_{1 \leq i \leq n}$.
\end{thmn}

Note that the original statement of Lemma~\ref{lemm:minimax} is subsumed by the above statement. We are ready to prove Lemma~\ref{lemm:minimax}.

\begin{proof}
Our construction closely follows~\citet{duchi2019lecture} and recent advances in proving minimax results for general losses; see for example~\citep{agarwal2009information}. To start, we fix an estimated denoiser $\hat{t}_n(x,z)=\hat{t}_n\p{x,z; (X_i,Z_i)_{1\leq i \leq n}}$ and define $\delta_{n, A,\sigma} = c^{1/2} \frac{\sigma^2}{\sigma^2 + A} \delta_n$, where $c$ is defined in Lemma~\ref{lemm:bayes_risk_mixture}. Next, focusing on one $v \in \mathcal{V}$, we get by Markov's inequality:

$$ \EE[m_v, A]{L(\hat{t}_n; m_v, A)}  \geq \delta_{n,A,\sigma}^2 \PP[m_v, A]{L(\hat{t}_n; m_v, A) \geq \delta_{n, A,\sigma}^2} $$
We next construct an estimator $\tilde{V}_n$ of $V_n$, namely we let:

$$ \tilde{V}_n = \argmin_{v \in \mathcal{V}}L(\hat{t}_n; m_v, A)$$
Notice that by Lemma~\ref{lemm:two_point_bound} and the assumption of the current Lemma, if the truth is $m_v$ and $L(\hat{t}_n; m_v, A) < \delta_{n, A,\sigma}^2$, then we definitely guessed correctly, in other words:

$$ L(\hat{t}_n; m_v, A) < \delta_{n, A,\sigma}^2 \; \Longrightarrow \; \tilde{V}_n = v $$
But taking the complements:

$$ \tilde{V}_n \neq v \; \Longrightarrow \;  L(\hat{t}_n; m_v, A) \geq \delta_{n, A,\sigma}^2$$
In terms of probabilities this implies that 
$$  \PP[m_v, A]{L(\hat{t}_n; m_v, A) \geq \delta_{n, A,\sigma}^2} \geq \PP[m_v, A]{ \tilde{V}_n \neq v}$$
Combining with our original result, and averaging over all $v$, we see that:
$$
\begin{aligned}
\sup_{m \in \mathcal{C}}\cb{\EE[m]{L(\hat{t}_n; m, A)}} &\geq \sup_{v \in \mathcal{V}_n}\cb{\EE[m_v,A]{L(\hat{t}_n; m_v, A)}} \\
 &\geq \frac{1}{\abs{\mathcal{V}_n}}\sum_{v \in \mathcal{V}_n} \EE[m_v, A]{L(\hat{t}_n; m_v, A)}\\
 &\geq \delta_{n,A,\sigma}^2 \cb{\frac{1}{\abs{\mathcal{V}_n}}\sum_{v \in \mathcal{V}_n} P_{m_v, A}\sqb{ \tilde{V}_n \neq v}} \\
&= \delta_{n,A,\sigma}^2 \PP{\tilde{V}_n \neq V_n} \\
&\geq \delta_{n,A,\sigma}^2 \inf_{\hat{V}_n}\PP{\hat{V}_n \neq V_n}
\end{aligned}
$$
Recall the definition of $\delta_{n, A,\sigma}^2$ and that $\hat{t}_n$ was arbitrary to conclude.
\end{proof}

\subsection{Proof of Lemma~\ref{lemm:bayes_risk_mixture}}
\label{subsec:to_testing_proof}

\begin{proof}
It only remains to prove~\eqref{eq:univariate_bayes_regret}. To this end, let us note that the result is essentially univariate; i.e. we may consider the following model: 

\begin{equation}
\label{eq:mixture_bayes_normal}
\begin{aligned}
&\mu \; \sim \; G = \frac{1}{2}\sqb{\nn\p{\eta_1, A} + \nn\p{\eta_2, A}} \\
&Z \mid \mu \sim \nn\p{\mu, \sigma^2}
\end{aligned}
\end{equation}

In this model, we want to prove that the Bayes risk of the Bayes estimator $t^*_G(Z) = \EE[G]{\mu \mid Z}$ satisfies the following inequality ($c>0, \Delta >0$): When $\abs{\eta_1 - \eta_2} \leq \Delta \sqrt{\sigma^2 + A}$ it holds that

\begin{equation}
\label{eq:univariate_bayes_regret_dropx}
\EE[G]{\p{t^*_{G}(Z) - \mu}^2} - \frac{A\sigma^2}{A+\sigma^2} \; \geq \; c \frac{\sigma^4}{\p{\sigma^2 + A}^2}\p{\eta_1 - \eta_2}^2
\end{equation}
The calculation is facilitated by Lemma~\ref{lem:fisher_info}, which states that $\EE[G]{\p{t^*_{G}(Z) - \mu}^2}=\sigma^2\sqb{1- \sigma^2 I(f_g)}$, where $f_g$ is the marginal density of $Z$ in~\eqref{eq:mixture_bayes_normal} and $I(f_g)$ is the Fisher information $\int \frac{f_g'(x)^2}{f_g(x)} dx$.

For the problem at hand, without loss of generality, we may take $\eta_1=0$, $\eta_2=\eta > 0$. Then the marginal distribution induced by $g$ is the mixture $\frac{1}{2}\sqb{ \nn\p{0, \sigma^2+A} + \nn\p{\eta, \sigma^2+A}}$, i.e. the pdf $f_g(\cdot)$ is:

$$ f_g(x) = \frac{1}{2 \sqrt{2\pi\p{\sigma^2+A}}}\sqb{ \exp\p{-\frac{x^2}{2(\sigma^2+A)}} + \exp\p{-\frac{(x-\eta)^2}{2(\sigma^2+A)}}}$$

$$ f_g'(x) =\frac{1}{2 \sqrt{2\pi\p{\sigma^2+A}}}\frac{1}{\sigma^2+A}\sqb{ -x\exp\p{-\frac{x^2}{2(\sigma^2+A)}} - (x-\eta)\exp\p{-\frac{(x-\eta)^2}{2(\sigma^2+A)}}}$$
Therefore, letting $\ell(u) = \exp(u)/(1+\exp(u))$ the logistic function, we see that:, 

$$
\begin{aligned}
\frac{f_g'(x)}{f_g(x)} &= \frac{1}{\sigma^2 + A} \frac{-x - (x-\eta)\exp\p{\frac{-\eta^2 + 2\eta x}{2(\sigma^2+A)}}}{1+ \exp\p{\frac{-\eta^2 + 2\eta x}{2(\sigma^2+A)}}} &=  \frac{1}{\sigma^2 + A}\sqb{-x + \eta \cdot \ell\p{\frac{-\eta^2 + 2\eta x}{2(\sigma^2+A)}}}
\end{aligned}
$$
Thus:

$$
\frac{f_g'(x)^2}{f_g(x)^2} = \frac{1}{\p{\sigma^2 + A}^2}\sqb{ x^2 + \eta^2  \cdot \ell^2\p{\frac{-\eta^2 + 2\eta x}{2(\sigma^2+A)}} - 2x\eta \cdot\ell\p{\frac{-\eta^2 + 2\eta x}{2(\sigma^2+A)}}}
$$
Then, letting $\tx = x/\sqrt{A+\sigma^2}$, $\teta = \eta/\sqrt{A+\sigma^2}$:

$$
\begin{aligned}
I(f_g)  &= \frac{1}{\p{\sigma^2 + A}^2}\frac{1}{2\sqrt{2\pi \p{\sigma^2+A}}} \int \cb{ x^2 + \eta^2  \cdot \ell^2\p{\frac{-\eta^2 + 2\eta x}{2(\sigma^2+A)}} - 2x\eta \cdot\ell\p{\frac{-\eta^2 + 2\eta x}{2(\sigma^2+A)}}} \\
   &\qquad\qquad\qquad\qquad\qquad\qquad\qquad\qquad\qquad          \cdot \p{ \exp\p{-\frac{x^2}{2(\sigma^2+A)}} + \exp\p{-\frac{(x-\eta)^2}{2(\sigma^2+A)}}} dx \\
 &= \frac{1}{2\sqrt{2\pi}\p{\sigma^2 + A}}\int \cb{ \tx^2 + \teta^2  \cdot \ell^2\p{\frac{-\teta^2 + 2\teta \tx}{2}} - 2\tx\teta \cdot\ell\p{\frac{-\teta^2 + 2\teta \tx}{2}}} \\ 
 & \qquad\qquad\qquad\qquad\qquad\qquad\qquad\qquad\qquad   \cdot \p{ \exp\p{-\frac{\tx^2}{2}} + \exp\p{-\frac{(\tx-\teta)^2}{2}}}d\tx \\
\end{aligned}
$$
Thus we may write $I(f_g)=\frac{1}{\sigma^2 + A}C(\teta)$, for some $C(\teta)$, which we now turn to study. Our first observation is that $C(0) = \EE{\tX^2} =1$ where $\tX \sim \nn\p{0,1}$. We claim that:

$$ C(\teta) = 1 - \frac{\teta^2}{4} + o(\teta^2) $$
To this end, we break up $C(\eta)$ into 6 components upon distributing terms, calling them $\IR_{0},\IIR_{0},\IIIR_{0}, \IR_{\teta}, \IIR_{\teta}, \IIIR_{\teta}$, where the subscript corresponds to integrating over $\tX \sim \nn\p{0,1}$ or $\tX \sim \nn\p{\teta,1}$.
$$
\begin{aligned}
&\IR_0 := \EE[0]{\tX^2} = 1, \;\; \IR_{\teta} := \EE[\teta]{\tX^2} = 1+\teta^2 \\
&\IIR_0 := \EE[0]{\teta^2\cdot \ell^2\p{\frac{-\teta^2 + 2\teta \tX}{2}}} = \frac{\teta^2}{4} + o(\teta^2) \;\;\;\;\;\;\;\text{ (dominated convergence theorem)}\\
&\IIR_{\teta} := \EE[\teta]{\teta^2\cdot \ell^2\p{\frac{-\teta^2 + 2\teta \tX}{2}}} = \frac{\teta^2}{4} + o(\teta^2)
\end{aligned}
$$
We may see the last result for example as follows, again using dominated convergence ($\teta \to 0$):

$$ \frac{\IIR_{\teta}}{\teta^2} =  \EE[\teta]{\ell^2\p{\frac{-\teta^2 + 2\teta \tX}{2}}} = \EE[0]{\ell^2\p{\frac{-\teta^2 + 2\teta (\tX + \teta)}{2}}} = \frac{1}{4} + o(1)$$

To bound $\IIIR$, it will be convenient to note that by Taylor's theorem it holds that $\ell(u) = \frac{1}{2} + \frac{u}{4} + O(u^3)$; in fact $\abs{\ell(u) - \frac{1}{2} - \frac{u}{4}} \leq \abs{u}^3$. Thus:

$$ \ell\p{\frac{-\teta^2 + 2\teta \tX}{2}} = \frac{1}{2}-\frac{\teta^2}{8} + \frac{\teta \tX}{4} + O\p{(-\teta^2 + 2\teta \tX)^3}$$, and so (one may check that again dominated convergence applies):

$$
\begin{aligned}
\IIIR_0 :&= -2\teta \EE[0]{\tX \ell\p{\frac{-\teta^2 + 2\teta \tX}{2}}} = -2\teta\EE[0]{\tX\p{\frac{1}{2}-\frac{\teta^2}{8} + \frac{\teta \tX}{4}} +  O\p{\tX(-\teta^2 + 2\teta \tX)^3}}\\
& = -2\teta\EE[0]{\tX\p{\frac{1}{2}-\frac{\teta^2}{8} + \frac{\teta \tX}{4}}} + o(\teta^2) = -2\teta\p{ 0 + \frac{\teta}{4}} + o(\teta^2) = -\frac{\teta^2}{2} + o(\teta^2)
\end{aligned}
$$

$$
\begin{aligned}
\IIIR_{\teta} :&= -2\teta \EE[\teta]{\tX \ell\p{\frac{-\teta^2 + 2\teta \tX}{2}}} = -2\teta\EE[\teta]{\tX\p{\frac{1}{2}-\frac{\teta^2}{8} + \frac{\teta \tX}{4}} +  O\p{\tX(-\teta^2 + 2\teta \tX)^3}}\\
&=  -2\teta\EE[0]{(\tX+\teta)\p{\frac{1}{2}-\frac{\teta^2}{8} + \frac{\teta (\tX+\teta)}{4}} +  O\p{(\tX+\teta)(-\teta^2 + 2\teta (\tX+\teta))^3}}\\
&= -2\teta\EE[0]{(\tX+\teta)\p{\frac{1}{2}-\frac{\teta^2}{8} + \frac{\teta (\tX+\teta)}{4}}} + o(\teta^2)\\
& = -2\teta\p{ \frac{\teta}{2} - \frac{\teta^3}{8} + \frac{\teta(\teta^2+1)}{4}} + o(\teta^2)\\
& = -\frac{3\teta^2}{2} + o(\teta^2)
\end{aligned}
$$
Add up to get :

$$ C(\teta) = \frac{1}{2}\sqb{\IR_0 + \IIR_0 + \IIIR_0 + \IR_{\teta} + \IIR_{\teta} + \IIIR_{\teta}} = 1 - \frac{\teta^2}{4} + o(\teta^2)$$
Then the regret is:

$$
\begin{aligned}
\sigma^2\sqb{1-\sigma^2 I(f_g)} - \frac{A\sigma^2}{\sigma^2+A} &= \sigma^2\p{1-\frac{\sigma^2}{\sigma^2+A} C(\teta)} - \frac{A\sigma^2}{\sigma^2+A} \\
&= \sigma^2\sqb{1-\frac{\sigma^2}{\sigma^2+A} \p{1 - \frac{\teta^2}{4} + o(\teta^2)}} - \frac{A\sigma^2}{\sigma^2+A}\\
& = \frac{\sigma^4}{\sigma^2 +A} \frac{1}{4} \teta^2 + \frac{\sigma^4}{\sigma^2 + A}o(\teta^2)\\
& = \frac{1}{4} \frac{\sigma^4}{\p{\sigma^2 +A}^2}\eta^2  + o(\teta^2)
\end{aligned}
$$

In particular, there exist $c >0, \Delta >0$ such that if $\teta \leq \Delta$:

$$ \sigma^2\sqb{1-\sigma^2 I(f_g)} - \frac{A\sigma^2}{\sigma^2+A} \geq c \frac{\sigma^4}{\p{\sigma^2 +A}^2}\eta^2$$
Recalling that $\teta =  \eta/\sqrt{A+\sigma^2}$, we conclude. We also note that we may let $c$ be arbitrarily close to $1/4$.

\end{proof}

\subsection{Proof and statement of Lemma~\ref{lem:fisher_info}}

\begin{lemm}
\label{lem:fisher_info}
Assume $\mu \sim g$ and $Z \mid \mu \;\sim\; \nn\p{\mu, \sigma^2}$. Also call $f_g$ the marginal density of $Z$ and define the Fisher information: 
$$I(f_g) := \int \frac{f_g'(x)^2}{f_g(x)} dx =  \EE[f_g]{\frac{f_g'(Z)^2}{f_g(Z)^2}}$$
Then it holds that:
$$\inf_{\hmu}\cb{\EE[g]{\p{\hmu-\mu}^2}}  = \sigma^2\sqb{1- \sigma^2 I(f_g)}$$
\end{lemm}

\begin{rema}
This formula is quite well know, see for example~\citep*{cohen2013empirical}. \citet{mukhopadhyay1995efficiency} call it Brown's formula in light of~\citep{brown1971admissible}. We give a proof for completeness; in which we do not justify switching integration and differentiation. For our purposes we only need the result for $g$ a mixture of two normals, in which case this is valid.
\end{rema}
\begin{rema}
As a simple application, consider $g= \nn\p{0, A}$, then $f_g = \nn\p{0, A+\sigma^2}$, so that $f_g(x) = \frac{1}{\sqrt{2\pi\p{\sigma^2+A}}} \exp\p{ -\frac{x^2}{2(\sigma^2+A)}}$ and $f_g'(x) = - \frac{1}{\sqrt{2\pi\p{\sigma^2+A}}}\frac{x}{\sigma^2+A}\exp\p{ -\frac{x^2}{2(\sigma^2+A)}}$. Thus $f_g'(x)^2/f_g(x)^2 = \frac{x^2}{\p{\sigma^2+A}^2}$ and $I(f_g) = \frac{1}{\sigma^2+A}$. The above result then states:
$$ \inf_{\hmu}\cb{\EE[g]{\p{\hmu-\mu}^2}} = \sigma^2\sqb{ 1- \frac{\sigma^2}{\sigma^2+A}} = \frac{\sigma^2 A}{\sigma^2 +A}$$
\end{rema}

\begin{proof}
We start with noting that the Bayes estimator is given by Tweedie's~\citep{efron2011tweedie} celebrated formula:
$$ \EE[g]{ \mu \mid Z = z} = z + \sigma^2 \frac{ f_g'(z)}{f_g(z)}$$
Then, the Bayes risk is given by (letting $\varepsilon := Z-\mu \sim \nn\p{0,\sigma^2}$):
$$
\begin{aligned}
\inf_{\hmu}\cb{\EE[g]{\p{\hmu-\mu}^2}}  &= \EE[g]{\p{\mu - Z - \sigma^2  \frac{ f_g'(Z)}{f_g(Z)}}^2} \\
&= \EE[g]{\p{-\varepsilon - \sigma^2  \frac{ f_g'(Z)}{f_g(Z)}}^2} \\
& = \sigma^2 + \sigma^4 I(f_g) + 2\sigma^2 \EE[g]{ \varepsilon \frac{ f_g'(Z)}{f_g(Z)}}\\
& = \sigma^2 - \sigma^4 I(f_g)
\end{aligned}
$$
It remains to justify that: $ \EE[g]{\varepsilon \frac{ f_g'(Z)}{f_g(Z)}} = -\sigma^2 I(f_g)$. To this end, first note that upon conditioning on $\mu$, we may use Stein's lemma, as follows:
$$
\begin{aligned}
\EE[g]{\varepsilon \frac{ f_g'(Z)}{f_g(Z)}} &= \EE[g]{\EE{\varepsilon \frac{ f_g'(\varepsilon + \mu)}{f_g(\varepsilon + \mu)} \mid \mu}}\\
&= \EE[g]{ \sigma^2 \EE{\frac{d}{d\varepsilon}\frac{ f_g'(\varepsilon + \mu)}{f_g(\varepsilon + \mu)} \mid \mu}}\\
&= \EE[g]{\sigma^2 \p{\frac{f_g''(Z)}{f_g(Z)} - \frac{f_g'(Z)^2}{f_g(Z)^2}}} \\
&= -\sigma^2 I(f_g)
\end{aligned}
$$
The last step that remains to be shown is that $\EE[g]{\frac{f_g''(Z)}{f_g(Z)}}=0$. But this is very similar to a standard Fisher information calculation, in which we interchange integration and differentiation to get that (here $\mu \sim g$): 
$$
\begin{aligned}
\EE[g]{\frac{f_g''(Z)}{f_g(Z)}} &= \int f_g''(z)dz = \frac{1}{\sqrt{2\pi \sigma^2}}\int \frac{d^2}{dz^2} \EE[g]{\phi((z - \mu)/\sigma)}dz = \frac{1}{\sqrt{2\pi \sigma^2}}\int \EE[g]{\frac{d^2}{dz^2}\phi((z - \mu)/\sigma)}dz\\
&= \frac{1}{\sqrt{2\pi \sigma^2}}\EE[g]{\int \frac{d^2}{dz^2}\phi((z - \mu)/\sigma)dz} = 0
\end{aligned}
$$

\end{proof}

\subsection{Local Fano's Lemma}

In this section we provide a Lemma to lower bound the expression $\smash{\inf_{\hat{V}_n} \PP{\hat{V}_n \neq V_n}}$ which appears in Lemma~\ref{lemm:minimax}. Below, we denote by $\mathbb P^X \otimes \nn\p{ m_v(\cdot),\; \sigma^2+A}$ the joint distribution of $(X,Z)$ when $X \sim \mathbb P^X$ and $Z \mid X \sim \nn\p{ m_v(X),\; \sigma^2+A}$.

\begin{lemm}[Local Fano]
\label{lemm:local_fano}

Assume there exists $\kappa >0$ such that for all $v,v' \in \mathcal{V}_n$:
$$ D_{\text{KL}}\p{ \mathbb P^X \otimes \nn\p{ m_v(\cdot),\; \sigma^2+A} \;\;||\;\; \mathbb P^X \otimes \nn\p{ m_{v'}(\cdot),\; \sigma^2+A}} \leq \kappa^2 $$
If also:
$$ \log( \abs{\mathcal{V}_n}) \geq 2(n\kappa^2 + \log(2))$$
Then:
$$ \inf_{\hat{V}_n} \PP{\hat{V}_n \neq V_n} \geq \frac{1}{2}$$
\end{lemm}

\begin{proof}
Let $V_n$ uniformly distributed on $\mathcal{V}_n$ and $\hat{V}_n$ any estimator of $V_n$. Then by Fano's inequality (Corollary 7.9 in \citet{duchi2019lecture}):

$$ \PP{ \hat{V}_n \neq V_n} \geq 1 - \frac{I(V_n;(X_i, Z_i)_{1 \leq i \leq n}) + \log(2)}{\log(\abs{\mathcal{V}_n)}}$$

Here $I(V_n;(X_i, Z_i)_{1 \leq i \leq n})$ is the mutual information between $V_n$ and $(X_i, Z_i)_{1 \leq i \leq n}$. 

Next fix $v,v' \in \mathcal{V}_n$ and let $P_v, P_{v'}$ the induced distributions of $(X_1,Z_1)$ induced by $m_v$, resp. $m_{v'}$ in model~\eqref{eq:np-fay-herriot-model}, then by (7.4.5) in~\citet{duchi2019lecture}:
$$ I(V_n;(X_i, Z_i)_{1 \leq i \leq n}) \leq \frac{1}{\abs{\mathcal{V}}^2}\sum_{v,v' \in \mathcal{V}_n} D_{\text{KL}}(P_v^n || P_{v'}^n) \leq n \max_{v, v' \in \mathcal{V}_n}D_{\text{KL}}(P_v || P_{v'})  \leq n\kappa^2$$
The result follows.
\end{proof}

\subsection{Fay Herriot results}

\begin{proof}

For the upper bound, we will use Theorem~\ref{theo:upper_bd}, where our regression estimator is just the ordinary least squares fit, i.e. $\hat{m}(x) = x^\top \hat{\beta}$ with $\hat{\beta} = (X^\top X)^{-1}X^{\top}Z_{1:n}$. By $X$ we mean the usual design matrix in which the vectors $X_1,\dotsc,X_n$ are stacked as rows into a matrix.

We start by decomposing the error:

$$
\begin{aligned}
\EE{ \p{\hat{m}(X_{n+1}) - m(X_{n+1})}^2} &= \EE{\p{X_{n+1}^\top \hat{\beta} - X_{n+1}^\top \beta}^2}\\
&=  \EE{\tr\p{(\hat{\beta}-\beta)^\top X_{n+1}X_{n+1}^\top (\hat{\beta}-\beta)}}\\
&= \EE{\tr\p{(\hat{\beta}-\beta)(\hat{\beta}-\beta)^\top X_{n+1}X_{n+1}^\top }}\\
&= \tr\p{\EE{(\hat{\beta}-\beta)(\hat{\beta}-\beta)^\top}\Sigma }
\end{aligned}
$$
Hence recalling that $\EE{\hat{\beta}}=\beta$, we only need to study the covariance of $\hat{\beta}$.

$$
\begin{aligned}
\Cov{\hat{\beta}} &= \EE{\Cov{\hat{\beta} \mid X_{1:n}}} \;\;+\;\; \Cov{\EE{\hat{\beta} \mid X_{1:n}}}\\
 &= (\sigma^2+A)\EE{ \p{X^\top X}^{-1}}\;\; + \;\; 0 \\
 &=(\sigma^2+A) \Sigma^{-1}\frac{1}{n-d-1}
\end{aligned}
$$
The last equality holds because $X^{\top}X$ follows a Wishart distribution. See Theorem 2 in \citet{rosset2018fixed} and references therein for similar results. In total we get:
$$\EE{ \p{\hat{m}(X_{n+1}) - m(X_{n+1})}^2} = \tr\p{(\sigma^2+A)\Sigma^{-1}\frac{1}{n-d-1} \Sigma} = \frac{d (\sigma^2 + A)}{n-p-1}$$
For the lower bound, we will apply Lemma~\ref{lemm:minimax}. First we let $\mathcal{V}_n$ be an $1/2$ packing of the  Euclidean ($\ell_2$) unit ball which has cardinality at least $2^d$ (Lemma 7.6. in~\citet{duchi2019lecture})

Then, for $v \in V_n$ we define $\theta_v = \varepsilon v$ (we will specify $\varepsilon$ later). Then we let $\beta_v = \Sigma^{-1/2}\theta_v$ and note that for two distinct $v,v'$:
$$
\begin{aligned}
\EE{ \p{ X_{n+1}^\top \beta_v - X_{n+1}^\top \beta_{v'}}^2} &= \tr\p{\EE{(\beta_v-\beta_{v'})(\beta_v-\beta_{v'})^\top}\Sigma } \\
&= \tr\p{\EE{\Sigma^{-1/2}(\theta_v-\theta_{v'})(\theta_v-\theta_{v'})^\top \Sigma^{-1/2}}\Sigma } \\
&= \EE{ \Norm{\theta_v-\theta_{v'}}^2_2} \\
&\geq \frac{\varepsilon^2}{4}
\end{aligned}
$$
In the last step we used the packing property of the set $\mathcal{V}_n$ we defined.

On the other hand:
$$
\begin{aligned}
&D_{\text{KL}}\p{\nn\p{0,\Sigma} \otimes \nn\p{ \langle \cdot,  \beta_v \rangle,\; \sigma^2+A} \;\;||\;\; \nn\p{0,\Sigma} \otimes \nn\p{\langle \cdot, \beta_{v'} \rangle,\; \sigma^2+A}} \\
= &\EE{ D_{\text{KL}}\p{ \nn\p{ X_{n+1}^\top \beta_v, \sigma^2+A} \;\;||\;\;  \nn\p{X_{n+1}^\top \beta_{v'}, \sigma^2+A} \cond X_{n+1}}} \\
= & \EE{ \frac{1}{2(A+\sigma^2)} \p{ X_{n+1}^\top \beta_v - X_{n+1}^\top \beta_{v'}}^2} \\
= & \frac{1}{2(A+\sigma^2)}\EE{ \Norm{\theta_v-\theta_{v'}}^2_2}\\
\leq & \frac{2\varepsilon^2}{A+\sigma^2}
\end{aligned}
$$
To apply Lemma~\ref{lemm:local_fano} we need the following to hold for a constant $C$:
$$ \log(2^d) \geq  C \frac{n \varepsilon^2}{A+\sigma^2}$$
So we may pick $\varepsilon^2 = c\frac{d(A+\sigma^2)}{n}$ for a constant $c$. Since $\varepsilon \to 0$ as $n\to\infty$, we may apply Lemma~\ref{lemm:minimax} for large enough $n$ with separation say $\varepsilon^2/10$, by which we can conclude.
\end{proof}

\subsection{Lipschitz results}
\begin{proof}
The upper bound follows from Theorem~\ref{theo:upper_bd}, where the regressor $\hat{m}_n$ is the $k$-nearest neighbor regression predictor (KNN) with optimally tuned number of neighbors, see Theorem 6.2 and Problem 6.7 in~\citet{gyorfi2006distribution}.

For the lower bound, we will apply Lemma~\ref{lemm:minimax}. To this end, we start by constructing~$\mathcal{V}_n$ as in the proof of Theorem 3.2. in~\citet{gyorfi2006distribution}: We define $M_n \in \mathbb N$ and partition $[0,1]^{d}$ (we will pick $M_n$ later) into $M_n^d$ cubes $A_{n,j}$ of side length $1/M_n$ and with centers $a_{n,j}$. Next we take any function $\bar{m}: \RR^{d} \to \RR$ which is $1$-Lipschitz, vanishes outside $[-\frac{1}{2}, \frac{1}{2}]^d$ and $C_I := \int \bar{m}^2(x)dx > 0$. We also define $\bar{m}_L(\cdot) = L \cdot \bar{m}(\cdot)$. Finally, for $j = 1,\dotsc, M_n^d$ we define:
$$ \bar{m}_{L, n, j}(x) = \frac{1}{M_n}\bar{m}_L(M_n(x-a_{n,j}))$$
Then we let $\mathcal{V}_n \subset \cb{\pm 1}^{M_n^d}$ with $\abs{\mathcal{V}_n} \geq \exp(M_n^d/8)$ and so that for all $v, v' \in \mathcal{V}_n$:
$$ \sum_{j=1}^{M_{n}^d} \ind\p{ v_j \neq v_j'} \geq \frac{M_{n}^d}{4}$$
Such a set exists by the Gilbert-Varshamov bound (Lemma 7.5 in~\citet{duchi2019lecture}). With $\mathcal{V}_n$ in hand, we define for $v \in \mathcal{V}_n$:
$$ m_{v}(x) = \sum_{j=1}^{M_n^d} v_j  \bar{m}_{L, n, j}(x) $$
We argue that $m_{v}(x)$ indeed is $L$-Lipschitz: All $\bar{m}_{L, n, j}$ are $L$-Lipschitz, since so is $\bar{m}_L$ and furthermore observe that all $\bar{m}_{L, n, j},\; j=1,\dotsc,M_n^d$ have disjoint support.

Next, take $v \neq v' \in \mathcal{V}_n$. Then, since the $\bar{m}_{L, n, j}$ have disjoint support:
$$
\begin{aligned}
\int \p{m_v(x) - m_{v'}(x)}^2 d\mathbb P^X(x) &= \sum_{j=1}^{M_n^d} (v_j - v_j')^2 \int  \bar{m}_{L, n, j}^2(x) d\mathbb P^X(x) \\
& \geq \sum_{j=1}^{M_n^d} (v_j - v_j')^2 \eta \int \bar{m}_{L, n, j}^2 dx\\
& =  \sum_{j=1}^{M_n^d} (v_j - v_j')^2 \frac{\eta L^2}{M_n^{2+d}} C_I \\
& = \frac{4 \eta L^2 }{M_n^{2+d}} C_I \sum_{j=1}^{M_n^d}\ind\p{ v_j \neq v_j'}  \\
& \geq \frac{4 \eta L^2}{M_n^{2+d}} C_I \frac{M_{n}^d}{4}  = \eta C_I \frac{L^2}{M_n^2}
\end{aligned}
$$
On the other hand, let us bound the KL divergence between the distributions induced by $m_v,m_{v'}$:
$$
\begin{aligned}
&D_{\text{KL}}\p{ \mathbb P^X \otimes \nn\p{ m_v(\cdot),\; \sigma^2+A} \;\;||\;\; \mathbb P^X \otimes \nn\p{ m_{v'}(\cdot),\; \sigma^2+A}} \\
= &\EE{ D_{\text{KL}}\p{ \nn\p{ m_v(X_{n+1}),\; \sigma^2+A} \;\;||\;\;  \nn\p{m_{v'}(X_{n+1}),\; \sigma^2+A} \mid X_{n+1}}}\\
= & \int \frac{1}{2(\sigma^2+A)} \p{ m_v(x) - m_{v'}(x)}^2 d\mathbb P^X(x) \\
\leq & \frac{1}{2\eta (\sigma^2+A)}  \int \p{ m_v(x) - m_{v'}(x)}^2 dx \\
\leq &\frac{1}{2\eta(\sigma^2+A)} \frac{4 L^2}{M_n^{2+d}} C_I \sum_{j=1}^{M_n^d}\ind\p{ v_j \neq v_j'}\\
\leq &\frac{2 C_I}{\eta (\sigma^2+A)}\frac{L^2}{M_n^2}
\end{aligned}
$$
Next, we will lower bound $\inf_{\hat{V}_n}\PP{\hat{V}_n \neq V_n}$ by Lemma~\ref{lemm:local_fano}. To get the condition, we need that for some $C>0$:
$$ M_n^{d} \geq C \frac{L^2 n}{(\sigma^2 +A)M_n^2}  \; \Leftrightarrow M_n \geq C\p{\frac{L^2 n}{\sigma^2+A}}^{\frac{1}{2+d}}$$

Hence for some $C$, we set $M_n =  \lceil C\p{\frac{L^2 n}{\sigma^2+A}}^{\frac{1}{2+d}} \rceil$. Then the separation between two hypotheses $m_v, m_{v'}$ is equal to (for another constant $C'$):
$$ \int \p{m_v(x) - m_{v'}(x)}^2 d\mathbb P^X(x) \geq \eta C_I\frac{L^2}{M_n^2} \geq C'\p{\frac{L^d(\sigma^2+A)}{n}}^{\frac{2}{2+d}}$$
We conclude by Lemma~\ref{lemm:minimax} upon noting that $M_n \to \infty$ and hence $\sup_{v \in \mathcal{V}_n} \sup_x \abs{m_v(x)} \to 0$ as $n \to \infty$.

\end{proof}

\section{Results for sample-split EB in Section~\ref{sec:sample_split_eb}}
in Section~\ref{sec:sample_split_eb} we made the following point: Even if we knew the true $A$, it would not be the optimal $A$ to plug into~\eqref{eq:plugin_hat_t}. We formalize this in the following proposition:

\begin{prop}
\label{prop:best_A}
Consider model~\eqref{eq:np-fay-herriot-model}. Fix any (deterministic) function $\tilde{m}:\mathcal{X} \to \RR$ and define:
\begin{equation}
\label{eq:optimal_A}
A_{\tilde{m}} := \EE[m,A]{\p{\tilde{m}(X_{n+1}) - Z_{n+1}}^2} - \sigma^2
\end{equation}
Then:
$$ \EE[m,A]{\p{t^*_{\tilde{m}, A_{\tilde{m}}}(X_{n+1},Z_{n+1}) - \mu_{n+1}}^2} = \inf_{\tilde{A} \geq 0} \EE[m,A]{\p{t^*_{\tilde{m}, \tilde{A}}(X_{n+1},Z_{n+1}) - \mu_{n+1}}^2} $$
The above expressions are equal to: $\frac{\sigma^2 A_{\tilde{m}}}{ \sigma^2 + A_{\tilde{m}}}$. Furthermore, a direct consequence is that:
$$ \EE[m,A]{\p{t^*_{\tilde{m}, A_{\tilde{m}}}(X_{n+1},Z_{n+1}) - \mu_{n+1}}^2} \leq \EE[m,A]{\p{t^*_{\tilde{m}, A}(X_{n+1},Z_{n+1}) - \mu_{n+1}}^2}$$
\end{prop}

\begin{proof}
Let us consider the following class of shrinkage rules, where $\lambda \in [0,1]$:

$$ t_{\lambda}(x,z) = \lambda \tilde{m}(x) + (1-\lambda)z = \lambda\p{\tilde{m}(x)-z} + z$$
Then our goal will be to minimize the following function over $\lambda \in [0,1]$:

\begin{equation}
\label{eq:lambda_loss}
J(\lambda) = \EE[m,A]{\p{t_{\lambda}(X_{n+1},Z_{n+1}) - \mu_{n+1}}^2}
\end{equation} 
To this end:
\begin{equation*}
\begin{aligned}
J(\lambda) &= \EE[m,A]{\cb{t_{\lambda}(X_{n+1},Z_{n+1}) - \mu_{n+1}}^2} \\
&= \EE[m,A]{\cb{\lambda \p{\tilde{m}(X_{n+1})-Z_{n+1}} + Z_{n+1} - \mu_{n+1}}^2} \\
&= \lambda^2 \EE[m,A]{\p{\tilde{m}(X_{n+1})-Z_{n+1}}^2} +2\lambda\EE[m,A]{\p{\tilde{m}(X_{n+1})-Z_{n+1}}\p{Z_{n+1} - \mu_{n+1}}} + \sigma^2 \\
&= \lambda^2 \EE[m,A]{\p{\tilde{m}(X_{n+1})-Z_{n+1}}^2} - 2\lambda \sigma^2 + \sigma^2
\end{aligned}
\end{equation*}
The last step follows from the two following intermediate results:

$$
\EE[m,A]{\tilde{m}(X_{n+1})\p{Z_{n+1} - \mu_{n+1}}} = \EE[m,A]{\tilde{m}(X_{n+1})\EE[m,A]{Z_{n+1} - \mu_{n+1} \mid X_{n+1}}} = 0
$$
$$
\EE[m,A]{Z_{n+1}\p{Z_{n+1} - \mu_{n+1}}} = \EE[m,A]{\Var[m,A]{ Z_{n+1} \mid \mu_{n+1}}} = \sigma^2
$$
We may now directly minimizer over $A$ to see that the optimal $\lambda$ is given by: 

$$ \lambda^*(\tilde{m}) = \frac{\sigma^2}{\EE[m,A]{\p{\tilde{m}(X_{n+1})-Z_{n+1}}^2}}$$
The form of $A_{\tilde{m}}$ then directly follows by noting the one-to-one correspondence $\lambda \leftrightarrow \frac{\sigma^2}{A_{\tilde{m}}+\sigma^2}$.
\end{proof}

We will now prove that for deterministic $\tilde{m}$, as in Proposition~\ref{prop:best_A}, parametric rates are possible in the estimation of $A_{\tilde{m}}$, which translate into $O(1/n)$ decay of the regret.

\begin{prop}
\label{prop:parametric_A}
Consider $n$ i.i.d. observations $(X_i, Z_i)$ from model~\eqref{eq:np-fay-herriot-model} with $A, \sigma >0$. Fix any (deterministic) function $\tilde{m}:\mathcal{X} \to \RR$ with $\EE{\tilde{m}(X_{n+1})^4} \leq M$ for some $M<\infty$ (here $X_{n+1} \sim \mathbb P^X$). Let:
$$ \hat{A}_n = \p{\frac{1}{n}\sum_{k=1}^{n} \p{\tilde{m}(X_{k})-Z_{k}}^2 - \sigma^2}_+$$
Then $\hat{t}_n= t^*_{\tilde{m}, \hat{A}_n}$ satisfies:
$$\EE[m,A]{L\p{\hat{t}_n; m, A}} \leq \EE[m,A]{L(t^*_{\tilde{m}, A_{\tilde{m}}}; m, A)} + O(1/n)$$
Thus also:
$$\EE[m,A]{L\p{\hat{t}_n; m, A}} \leq \EE[m,A]{L(t^*_{\tilde{m},A}; m, A)} + O(1/n)$$
\end{prop}

\begin{proof}
We consider again the $J(\lambda)$ from \eqref{eq:lambda_loss} and recall that $J(\lambda^*(\tilde{m})) = \min_{\lambda \geq 0 } J(\lambda)$. We note that $J(\lambda)$ is a convex quadratic in $\lambda$ with:

$$ J'(\lambda) = 2\lambda \EE[m,A]{\p{\tilde{m}(X_{n+1})-Z_{n+1}}^2} - 2\sigma^2,\;\; J''(\lambda) = 2\EE[m,A]{\p{\tilde{m}(X_{n+1})-Z_{n+1}}^2}$$
Thus, since $J'(\lambda^*(\tilde{m}))  = 0$, we get for any $\lambda$:

$$J(\lambda) = J(\lambda^*(\tilde{m})) + \EE[m,A]{\p{\tilde{m}(X_{n+1})-Z_{n+1}}^2} \p{ \lambda - \lambda^*(\tilde{m})}^2$$
This means that:
$$ L(t_{\lambda}; m, A) = L(t_{\lambda^*(\tilde{m})}; m, A) + \EE[m,A]{\p{\tilde{m}(X_{n+1})-Z_{n+1}}^2}\p{ \lambda - \lambda^*(\tilde{m})}^2$$

Hence to conclude we will need to bound $\EE[m,A]{\p{ \hat{\lambda}_n - \lambda^*(\tilde{m})}^2}$, where:
$$\hat{\lambda}_n = \frac{\sigma^2}{ \sigma^2 \lor \p{\frac{1}{n}\sum_{k=1}^{n} \p{\tilde{m}(X_{k})-Z_{k}}^2}}$$

%Let $\tilde{Z}_i = Z_i/\sqrt{A+\sigma^2}$ and $\widetilde{M}_i = \tilde{m}(X_i)/\sqrt{A+\sigma^2}$, then: 
%\begin{equation}
%\begin{aligned}
%\EE{ \p{\hat{\lambda} - \lambda^*(\tilde{m})}^2} & = \frac{\sigma^4 n^2}{(A+\sigma^2)^2} \EE{\p{\frac{1}{\sum_{k=1}^n \p{\tilde{Z}_{n+k}-\widetilde{M}_{n+k}}^2} - \frac{1}{n\EE{\p{\tilde{Z}_{n+k}-\widetilde{M}_{n+k}}^2}}}^2} \\
%& \leq \frac{\sigma^4 n^2}{(A+\sigma^2)^2}\EE{\p{\frac{1}{\sum_{k=1}^n \p{\tilde{Z}_{n+k}}^2} - \frac{1}{n}}^2 }
%\end{aligned}
%\end{equation}
Using the fact that both $\sigma^2 \lor \p{\frac{1}{n}\sum_{k=1}^{n} \p{\tilde{m}(X_{k})-Z_{k}}^2}$ and $\EE[m,A]{\p{\tilde{m}(X_{k}) - Z_k}^2}$ are $\geq \sigma^2$ and Taylor's theorem applied to $u \mapsto 1/u$, we get:

\begin{equation*}
\begin{aligned}
\EE[m,A]{ \p{\hat{\lambda} - \lambda^*(\tilde{m})}^2} &= \sigma^4 \EE[m,A]{ \p{\frac{1}{ \sigma^2 \lor \p{\frac{1}{n}\sum_{k=1}^{n} \p{\tilde{m}(X_{k})-Z_{k}}^2}} -  \frac{1}{\EE[m,A]{\p{\tilde{m}(X_{k}) - Z_k}^2}}}^2} \\
\leq &\EE[m,A]{ \p{ \sigma^2 \lor \p{\frac{1}{n}\sum_{k=1}^{n} \p{\tilde{m}(X_{k})-Z_{k}}^2} -  \EE[m,A]{\p{\tilde{m}(X_{k}) - Z_k}^2}}^2} \\
\leq & \EE[m,A]{ \p{ \frac{1}{n}\sum_{k=1}^{n} \p{\tilde{m}(X_{k})-Z_{k}}^2 -  \EE[m,A]{\p{\tilde{m}(X_{k}) - Z_k}^2}}^2} \\
= & \Var[m,A]{ \frac{1}{n} \sum_{k=1}^n \p{\tilde{m}(X_{k})-Z_{k}}^2}\\
= & \frac{1}{n}  \Var[m,A]{\p{\tilde{m}(X_{k})-Z_{k}}^2}
\end{aligned}
\end{equation*}

This is $O(1/n)$ as long as $\Var[m,A]{\p{\tilde{m}(X_{k})-Z_{k}}^2}$ is upper bounded, which is the case under the given assumptions. The last statement follows from Proposition~\ref{prop:best_A}.
%We first study the squared error in estimating $\EE[m]{\p{\tilde{m}(X_{n+1})-Z_{n+1}}^2}$.

%$$\EE{\cb{\frac{1}{n}\sum_{k=1}^{n} \p{\tilde{m}(X_{n+k})-Z_{n+k}}^2 - \EE{\p{\tilde{m}(X_{n+1})-Z_{n+1}}^2}}^2} = \frac{1}{n}\Var{\p{\tilde{m}(X_{n+1})-Z_{n+1}}^2} $$
%Next:

%$$
%\begin{aligned}
%\Var{\p{\tilde{m}(X_{n+1})-Z_{n+1}}^2} &\leq \EE{\p{\tilde{m}(X_{n+1})-Z_{n+1}}^4} \\
%&\leq 8\EE{\p{\tilde{m}(X_{n+1})-m(X_{n+1})}^4} + 8\EE{ \p{Z_{n+1} - m(X_{n+1})}^4} \\
%& \leq 8\EE{\p{\tilde{m}(X_{n+1})-m(X_{n+1})}^4} + 24(A+\sigma^2)^2
%\end{aligned}
%$$

%\fixme{Can make this tighter by considering the characterization of chi-squared as mixture of Poissons, not sure I will have time for this though.}

\end{proof}

We are now in a position to prove Theorem~\ref{theo:optimal_A_datadriven}

\begin{proof}[Theorem~\ref{theo:optimal_A_datadriven}]
We apply Proposition~\ref{prop:parametric_A} for the data in fold $I_2$ conditionally on the first fold, i.e. conditionally on $Z_{I_1}, \mu_{I_1}, X_{I_1}$.
\end{proof}

\section{Results under misspecification}

\subsection{Proof of Theorem~\ref{theo:james_stein} (James-Stein property)}

Before proceeding with the proof, let us introduce the following lemma:

\begin{lemm}
\label{lemm:js_loc_shift}
Fix $\nu \in \mathbb N$, a fixed vector $\xi = (\xi_1,\dotsc,\xi_{\nu})$, a mean vector $\theta=(\theta_1,\dotsc, \theta_{\nu})$ and independent $Y_1, \dotsc, Y_{\nu}$ distributed as $Y_i \sim \nn\p{\theta_i, \sigma^2}$. Then consider the following positive-part James-Stein type estimator, parametrized by $a >0$:

\begin{equation}
\label{eq:green_strawderman_js}
\hat{\theta}_{a} = \xi + \p{1 - \frac{a \sigma^2}{\Norm{Y-\xi}_2^2}}_+\p{ Y - \xi}
\end{equation}
This estimator has risk:

\begin{equation}
\label{eq:green_strawderman_risk}
\EE{\Norm{\hat{\theta}_a - \theta}^2_2} \leq \nu \sigma^2 - a\sigma^2\sqb{2(\nu -2) -a }\EE{\frac{\sigma^2}{\Norm{Y-\xi}_2^2}}
\end{equation}
In particular, if $\nu \geq 5$ (resp. $\nu \geq 3$), $\hat{\theta}_{\nu}$ (resp. $\hat{\theta}_{\nu -2}$) has squared error risk $< \nu \sigma^2$.
\end{lemm}

\begin{proof}
Estimator~\eqref{eq:green_strawderman_js} where we do not take the positive part of $\p{1 - \frac{a \sigma^2}{\Norm{Y-\xi}_2^2}}_+$ has risk precisely equal to the RHS in~\eqref{eq:green_strawderman_risk}. This is well known, see for example Lemma 1 in~\citep{green1991james} and references therein. The positive part estimator then has smaller risk, as also follows from well known results on James-Stein    estimation, see e.g.~\citep{baranchik1964multiple}. Finally when $a = \nu \geq 5$,  $a\sigma^2\sqb{2(\nu -2) -a } = \sigma^2 \sqb{\nu - 4} > 0$.
\end{proof}

We are ready to prove Theorem~\ref{theo:james_stein}:

\begin{proof}
Let $\aPP[I_1]{\; \cdot \;} = \PP{ \;\cdot \mid Z_{I_1}, \mu_{1:n}, X_{1:n}}$. Then w.r.t. $\aPP[I_1]{\cdot}$, it holds that $(Z_i)_{i \in I_2}$ are independent and $Z_i \sim \nn\p{\mu_i, \sigma^2}$ for $i \in I_2$. Furthermore $\hat{m}_{I_1}(X_{I_2})=(\hat{m}_{I_1}(X_i))_{i \in I_2}$ is deterministic w.r.t. $\aPP[I_1]{\cdot}$ and also recall that:

$$
\begin{aligned}
\hat{\mu}_{I_2}^{\text{EBCF}} &= \frac{\sigma^2}{\hat{A}_{I_2}+ \sigma^2}\hat{m}_{I_1}(X_{I_2}) +\frac{\hat{A}_{I_2}}{\hat{A}_{I_2}+ \sigma^2}Z_{I_2} \\
&= \hat{m}_{I_1}(X_{I_2}) + \p{ 1-\frac{\sigma^2}{\hat{A}_{I_2}+ \sigma^2}}\p{ Z_{I_2} - \hat{m}_{I_1}(X_{I_2})}
\end{aligned}
$$

Also from~\eqref{eq:A_sure_homoskedastic} it holds that: 

$$ \hat{A}_{I_2} = \p{\frac{1}{\abs{I_2}}\sum_{i \in I_2} \p{\hat{m}_{I_1}(X_{i})-Z_i}^2 - \sigma^2}_+ $$
Thus $\hat{\mu}_{I_2}$ takes precisely the form from~\eqref{eq:green_strawderman_js} with $a = \abs{I_2}$ and thus applying Lemma~\ref{lemm:js_loc_shift} (w.r.t. $\aPP[I_1]{\; \cdot \;}$, also by assumption $\abs{I_2} \geq 5$), we get:

 $$ \sum_{i \in I_2} \aEE[I_1]{ \p{\mu_i - \hat{\mu}_i^{\text{EBCF}}}^2} <  \sum_{i \in I_2} \aEE[I_1]{ (\mu_i - Z_i)^2} = \abs{I_2} \sigma^2 $$
Integrate w.r.t. $Z_{I_1}$, to get:

$$ \sum_{i \in I_2} \EE{ \p{\mu_i - \hat{\mu}_i^{\text{EBCF}}}^2 \cond \mu_{1:n}, X_{1:n}} < \abs{I_2} \sigma^2 $$
Now apply the symmetric argument with the folds flipped to also get:

$$ \sum_{i \in I_1} \EE{ \p{\mu_i - \hat{\mu}_i^{\text{EBCF}}}^2 \cond \mu_{1:n}, X_{1:n}} < \abs{I_1} \sigma^2 $$
Add both inequalities and divide by $n$ to conclude.
\end{proof}

\subsection{SURE results}
\label{sec:SURE_results}
Below we prove Theorem~\ref{theo:sure_equal_variance}. Throughout the proof we deal with the more general case of unequal variances. In particular, we replace the assumption that $(X_i,Z_i)$ satisfy~\eqref{eq:np-fay-herriot-compound-model-variance} by the following model (while keeping all other assumptions):

\begin{equation*}
(X_i, \mu_i) \sim \mathbb{P}^{(X_i,\mu_i)},\;\; Z_i \mid \mu_i, X_i \sim \p{\mu_i, \sigma_i^2},\;
\text{ i.e. }\EE{Z_i \mid \mu_i, X_i} = \mu_i, \; \Var{Z_i \mid \mu_i, X_i} = \sigma_i^2
\end{equation*}

\begin{proof}[Theorem~\ref{theo:sure_equal_variance}] Our proof closely follows~\citet{xie2012sure}. Let $n_2 = \abs{I_2}$. We also use the same notation as in the proof of Theorem~\ref{theo:james_stein}, wherein $\aPP[I_1]{\; \cdot \;} = \PP{ \;\cdot \mid Z_{I_1}, \mu_{1:n}, X_{1:n}}$. For $i \in I_2$ we also write $\tm_i = \hat{m}_{I_1}(X_i)$. We rewrite the $\SURE$ expression as follows:
$$ \SURE_{I_2}(A) = \frac{1}{n_2}\sum_{i \in I_2} \p{\sigma_i^2 + \frac{\sigma_i^4}{(A+\sigma_i^2)^2}(Z_i - \tm_i)^2 - 2\frac{\sigma_i^4}{ A + \sigma_i^2}} = \frac{1}{n_2}\sum_{i \in I_2} \sqb{\frac{\sigma_i^4}{(A+\sigma_i^2)^2}(Z_i-\tm_i)^2 + \frac{\sigma_i^2(A-\sigma_i^2)}{A + \sigma_i^2}}$$

We also define $\ell_{I_2}(A)$, the average loss in fold $I_2$ when we estimate $\mu_i$ by $t^*_{\hat{m}_{I_{1}},A}(X_i, Z_i)$, i.e.:
$$ \ell_{I_2}(A) := \frac{1}{n_2}\sum_{i \in I_2}\p{ \mu_i - t^*_{\hat{m}_{I_{1}},A}(X_i, Z_i)}^2$$
Next we collect the difference between the SURE risk estimate and the actual loss:
{\small
$$
\begin{aligned}
\SURE_{I_2}(A) - \ell_{I_2}(A) &= \frac{1}{n_2}\sum_{i \in I_2} \sqb{\p{\frac{\sigma_i^4}{(A+\sigma_i^2)^2}(Z_i-\tm_i)^2 + \frac{\sigma_i^2(A-\sigma_i^2)}{A + \sigma_i^2}} - \p{\mu_i - \frac{A}{A + \sigma_i^2}Z_i - \frac{\sigma_i^2}{A+\sigma_i^2}\tm_i}^2} \\
&=\frac{1}{n_2}\sum_{i \in I_2}\cb{ \sqb{\p{Z_i-\tm_i}^2 - \sigma_i^2 - \p{\mu_i-\tm_i}^2} - \frac{2A}{A + \sigma_i^2}\sqb{\p{Z_i - \tm_i}^2 - (Z_i-\tm_i)(\mu_i-\tm_i) - \sigma_i^2}} \\
& = \text{I} + \text{II}
\end{aligned}
$$
}

We consider each term independently. The first term does not depend on $A$, hence is easier to study.
$$
\begin{aligned}
&\aEE[I_1]{\abs{\sum_{i \in I_2} \sqb{\p{Z_i-\tm_i}^2 - \sigma_i^2 - \p{\mu_i-\tm_i}^2}}}^2 \\
\leq \; &\aEE[I_1]{\p{\sum_{i \in I_2} \sqb{\p{Z_i-\tm_i}^2 - \sigma_i^2 - \p{\mu_i-\tm_i}^2}}^2}\\
= \; &\sum_{i \in I_2} \aVar[I_1]{\p{Z_i-\tm_i}^2} \\
\leq \; & 8\sum_{i \in I_2} \p{ \aEE[I_1]{Z_i^4} + \tm_i^4}
\leq 8n_2 \p{ \Gamma^4 + M^4}
\end{aligned}
$$

The second term depends on $A$ and we want a result that is uniform in $A$. Without loss of generality, we may assume that the indices in $I_2 = \cb{i_{1}, i_{2},\dotsc}$ are arranged such that $\sigma^2_{i_{1}} \leq \sigma^2_{i_{2}} \leq ...$ (otherwise we may just rearrange). Then, as observed in \citet{li1986asymptotic, xie2012sure}:

$$
\begin{aligned}
 &\sup_{0\leq A \leq \infty} \abs{\sum_{i \in I_2} \frac{A}{A + \sigma_i^2}\sqb{ \p{Z_i-\tm_i}^2 - (Z_i-\tm_i)(\mu_i - \tm_i) - \sigma_i^2}} \\
\leq & \sup_{0 \leq c_n \leq \dotsc \leq c_1 \leq 1} \abs{\sum_{i \in I_2} c_i \sqb{\p{Z_i-\tm_i}^2 -(Z_i-\tm_i)(\mu_i - \tm_i) - \sigma_i^2}}\\
= & \max_{j=1,\dotsc,n_2} \bigg |\underbrace{\sum_{k=1}^j\sqb{\p{Z_{i_k} - \tm_{i_k}}^2 - (Z_{i_k} - \tm_{i_k})(\mu_{i_k} - \tm_{i_k}) - \sigma_{i_k}^2}}_{=:M_j}\bigg |
\end{aligned}
$$
Next notice that $M_j, j=1,\dotsc, n_2$ is a martingale w.r.t. $\aPP[I_1]{\cdot}$, so by the $L^2$ maximal inequality, for a constant $C>0$:
$$
\begin{aligned}
\aEE[I_1]{\max_{j=1,\dotsc,n_2} M_j^2} \leq 4 \aEE[I_1]{M_{n_2}^2} &= 4 \sum_{i \in I_2} \aVar[I_1]{\p{Z_i - \tm_i}^2  - (Z_i - \tm_i)(\mu_i - \tm_i) - \sigma_i^2} \\
&\leq C n_2(\Gamma^4 + M^4)
\end{aligned}
$$
The results together imply that for a constant $C' > 0$:
$$ \aEE[I_1]{\sup_{A \geq 0} \abs{\SURE_{I_2}(A) - \ell_{I_2}(A)}} \leq C'\sqrt{\Gamma^4 + M^4}\frac{1}{\sqrt{n_2}}$$
But by definition of $\hat{A}_{I_2}$, $\SURE_{I_2}(\hat{A}_{I_2}) \leq \inf_{A \geq 0}\SURE_{I_2}(A)$ and so for any $A \geq 0$:

$$ \aEE[I_1]{\ell_{I_2}(\hat{A}_{I_2})} \leq \aEE[I_1]{\ell_{I_2}(A)} + \aEE[I_1]{\sup_{A \geq 0} \abs{\SURE_{I_2}(A) - \ell_{I_2}(A)}} \leq \aEE[I_1]{\ell_{I_2}(A)}\; + \; C'\sqrt{\Gamma^4 + M^4}\frac{1}{\sqrt{n_2}}$$
This holds for any $A \geq 0$, hence it remains valid after we take the infimum over $A \geq 0$.
\end{proof}

\subsection{Proof of Corollary~\ref{coro:equal_sample_size_compound}}

\begin{proof}
By Theorem~\ref{theo:sure_equal_variance}:

{\footnotesize
$$
\begin{aligned}
\frac{2}{n}\sum_{i \in I_2} \EE{ (\mu_i - \hat{\mu}_i^{\text{EBCF}})^2 \mid X_{1:n}, \mathbf{\mu}_{1:n}, Z_{I_{1}}} \leq &\inf_{A \geq 0}\cb{\frac{2}{n}\sum_{i \in I_2}  \EE{ \p{\mu_i - t^*_{\hat{m}_{I_{1}},A}(X_i, Z_i)}^2 \mid X_{1:n}, \mathbf{\mu}_{1:n}, Z_{I_{1}}}} \; +\; O\p{ \frac{1}{\sqrt{n}}}
\end{aligned}
$$}
Next integrate over $X_{1:n}, \mathbf{\mu}_{1:n}, Z_{I_{1}}$ and pull the $\inf$ outside of the expectation and use the fact that $(X_i, Z_i, \mu_i)$ are i.i.d. to get for fresh $(X_{n+1}, Z_{n+1})$:

$$ \frac{2}{n}\sum_{i \in I_2} \EE{ (\mu_i - \hat{\mu}_i^{\text{EBCF}})^2} \leq \inf_{A \geq 0}\cb{\EE{ \p{\mu_{n+1} - t^*_{\hat{m}_{I_1},A}(X_{n+1}, Z_{n+1})}^2 }} \; +\; O\p{ \frac{1}{\sqrt{n}}}$$
Then, make the choice $A =  \EE{ \p{\hat{m}_{I_1}(X_{n+1}) - \mu_{n+1}}^2}$ to get:

$$ \frac{2}{n}\sum_{i \in I_2} \EE{ (\mu_i - \hat{\mu}_i^{\text{EBCF}})^2} \leq \frac{\sigma^2\EE{ \p{\hat{m}_{I_1}(X_{n+1}) - \mu_{n+1}}^2}}{\sigma^2 + \EE{ \p{\hat{m}_{I_1}(X_{n+1}) - \mu_{n+1}}^2}}\; +\; O\p{ \frac{1}{\sqrt{n}}}$$
Repeat the same argument with $I_1,I_2$ flipped, add the results and divide by $2$ to conclude.
\end{proof}

\end{appendix}

\end{document}